\documentclass[preprint,12pt]{article}

\usepackage{amssymb}
\usepackage{amsmath,amsfonts,amssymb,amsthm,amscd}
\usepackage{graphicx}
\usepackage{textcomp}
\usepackage{color}
\usepackage{booktabs, multirow}
\usepackage{siunitx}
\usepackage{caption}

\newcommand{\tvv}{\boldsymbol{\theta}}

\newcommand{\E}{\mathbb{E}}
\newcommand{\PP}{\mathbb{P}}

\newcommand{\cov}{\mbox{\textrm{cov}}}

\newtheorem{Theorem}{Theorem}[section]

\newtheorem{Proposition}[Theorem]{Proposition}

\newtheorem{Remark}[Theorem]{Remark}
\newtheorem{Example}[Theorem]{Example}

\begin{document}


\title{Credit risk for large portfolios of green and brown loans:  extending the ASRF model.}

\author{A. Ramponi and S. Scarlatti \\ Department of Economics and Finance \\ University of Roma Tor Vergata}

\maketitle

\begin{abstract}
We propose a credit risk model for portfolios composed of green and brown loans, extending the ASRF framework via a two-factor copula structure. Systematic risk is modeled using potentially skewed distributions, allowing for asymmetric creditworthiness effects, while idiosyncratic risk remains Gaussian. Under a non-uniform exposure setting, we establish convergence in quadratic mean of the portfolio loss to a limit reflecting the distinct characteristics of the two loan segments. Numerical results confirm the theoretical findings and illustrate how value-at-risk is affected by portfolio granularity, default probabilities, factor loadings, and skewness. Our model accommodates differential sensitivity to systematic shocks and offers a tractable basis for further developments in credit risk modeling, including granularity adjustments, CDO pricing, and empirical analysis of green loan portfolios. 
\bigskip


\end{abstract}


\section{Introduction}
The Vasicek loan portfolio loss model calculates the distribution of the total loss of a large portfolio of loans by aggregating the losses from individual loans. These individual loan losses depend on the borrower default probabilities, LGDs, and the correlation between the defaults. The first time the model was presented by O. Vasicek \cite{vas1}, it was considered as a breakthrough, and had a profound impact on subsequent developments of regulatory aspects of credit risk. Nowadays, almost forty years later, Vasicek intuition remains prominent, however many changes have occurred in the credit universe. One of the most notable examples is the increasing growth of the so-called 'green loans', that is, loans that are targeted to finance projects or initiatives that have positive environmental benefits or contribute to sustainability. Despite their importance for society, green loans can default exactly as the regular ones, often called "brown loans". For instance, the project for the construction of a large solar farm may experience unexpected technical issues or bureaucratic obstacles from local municipalities that strongly delay its proper functioning, making it difficult to meet the declared sustainability and energy targets and loan repayment obligations. All in all, the reason for a green loan to default can eventually be different than that of a brown one, but the conclusion stays the same. It is therefore a natural research question to ask if the Vasicek model can accommodate in a coherent way both types of loans at the same time. 
This question is not at all an academic one, by the fact that, in recent times, the loan portfolios of many financial institutions around the world have become of mixed type just in the sense we have described. The increased attention of society towards our environment and sustainable developments boosted the spread of green loans on credit markets. At the same time, their specific appeal was propelled by recent advancements in legal and regulatory frameworks aiming to clarify their definitions and usage. At this point, it is useful to explicitly report the LMA definition of these instruments:

\smallskip

\textit{Green loans are any type of loan instrument made available exclusively to finance or re-finance, in whole or in part, new and/or existing eligible Green Projects.}\footnote{Loan Market Association (LMA) Green Loan Principles.}

\smallskip

A more quantitative insight into their temporal diffusion in the markets is provided in Figure \ref{fig:loan_stats}, which plots the annual total number and value of green loans, based on data from \textit{Environmental Finance Data} \texttt{(efdata.org)}.
 
\begin{figure}[h]
\centering
\includegraphics[width=16cm,height=8cm]{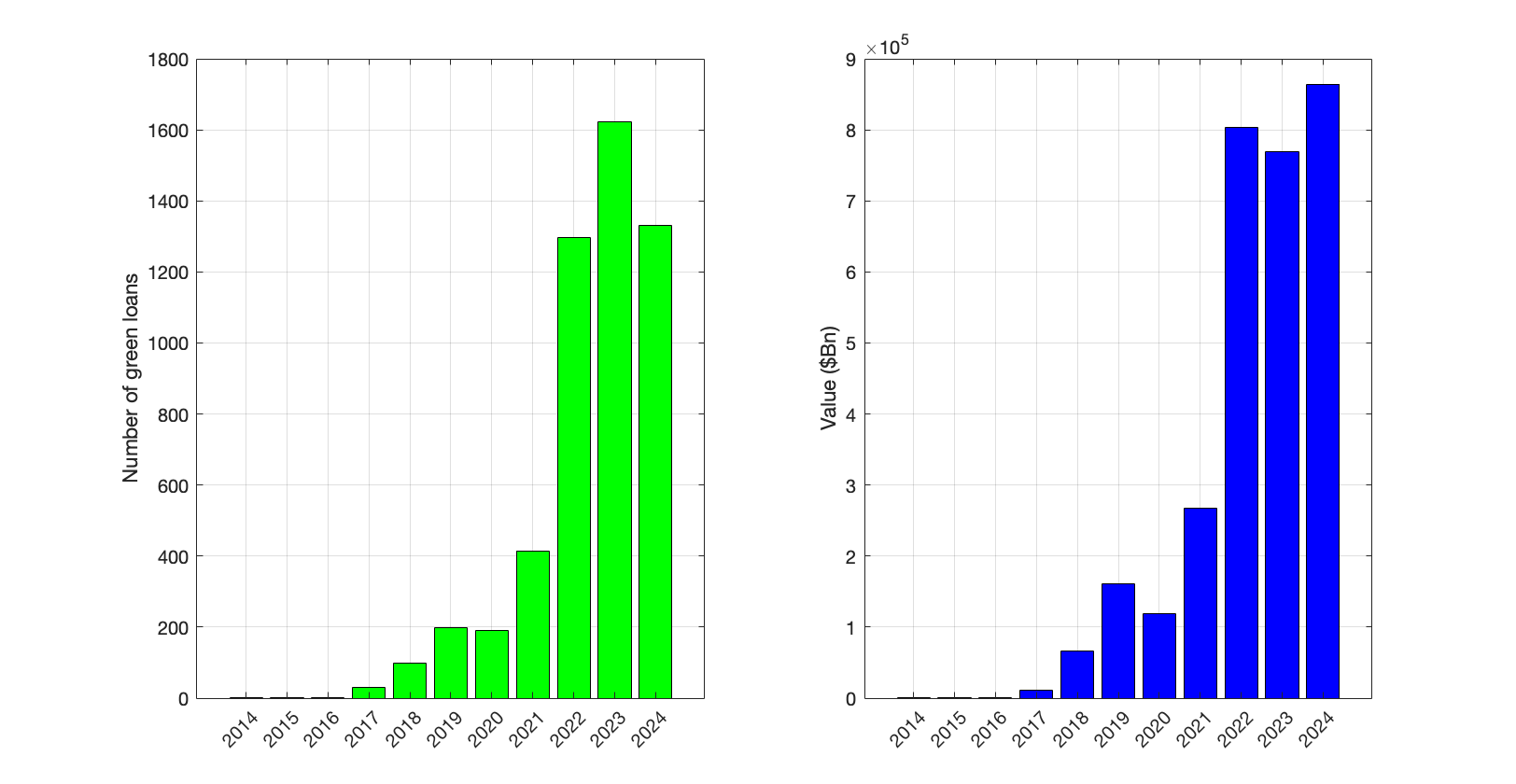}
\caption{Annual number (left plot) and value (right plot) in dollar billions of green loans (source: Environmental Finance Data, \texttt{efdata.org}).}
\label{fig:loan_stats}
\end{figure} 

We wish to remark that our paper is not the first to explore how the consideration of green loans in a large loan portfolio introduces modification in the Vasicek approach and related distributional results. Indeed, Agliardi \cite{agl21} already considered this possibility. More specifically, in that paper the securitization process is discussed creating CDO tranches that have a green-brown loan portfolio as collateral. The author assumes that the CDO senior tranche is sold in the bond market and the proceeds used to finance new green projects, naming the whole procedure "green securitization". Another direction of research close to the one presented here is  considered in Lee and Poon \cite{lp15}. There, the authors introduce a useful generalization of Vasicek model based on the consideration of a systematic factor affecting all borrower returns which is skew-normal distributed in the sense of Azzalini \cite{azz1}. This introduces a further parameter to the distribution of borrower returns which adds wider flexibility to the Vasicek approach still preserving its tractability. Lee and Poon test their model on the loan portfolios of a sample of $15$ different Korean banks coming to the conclusion that, in most of the cases, their model produces a better fit of the observed data w.r.t. the canonical gaussian model. In the present paper we  build a coherent unified framework for the results of Agliardi \cite{agl21} and of Lee and Poon \cite{lp15} and show that it can be quite effective in analyzing the credit risk profiles of large portfolios. 
Indeed, modeling the default structure of credit portfolios is one of the major challenges in such a framework, handling a model which skips unnecessary complexities, however, catches most relevant aspects. The \textit{Large Homogeneous Portfolios} (LHP) approach (see \cite{gordy}), which encompasses that of Vasicek,  has these characteristics. 
It is a risk assessment methodology that is used to evaluate the credit risk of large portfolios composed of similar loans. In the limit of an infinitely granular homogeneous portfolio, the total portfolio loss becomes a deterministic function when conditioned on a single systematic risk factor, which is the fundamental insight underpinning the Vasicek model. As a consequence,  Gordy \cite{gordy} showed that the value-at-risk of a portfolio could be well approximated by the quantile of an explicit function of the conditional expectation of losses. This result forms the theoretical basis for the Internal Ratings-Based (IRB) approach to credit risk capital under the Basel II Accord, and continues to influence regulatory capital modeling in Basel III and beyond.

Since then, a rich body of literature has extended the ASRF framework in multiple directions. Emmer and Tasche \cite{EmmTasc} and Martin and Wilde \cite{MartWild} explored granularity adjustments to account for real portfolios having a finite number of loans, highlighting deviations from the asymptotic approximation. Pykhtin \cite{pykhtin2004} and Frey and McNeil \cite{frey2003} introduced multi-factor extensions and copula-based approaches to capture sectoral or regional dependencies, while maintaining computational efficiency. The implications of using the ASRF framework from a regulatory perspective are considered, e.g., by Casellina et al. (see \cite{clu23} and the references therein), who investigated the impact of the statistical estimation error of the (long-run) probability of default (PD) of a borrower, proposing a practical correction method. 
To address non-Gaussian features in empirical credit loss distributions, Schloegl and O’Kane \cite{so05} proposed the use of a Student-t copula, and provide analytical formulae for the density and the cdf of the portfolio loss distribution, and they compared the value-at-risk implied by such a distributional assumption. Similarly, as previously mentioned, Lee and Poon \cite{lp15} introduced skewed distributions for the systematic factor, capturing asymmetries in borrower creditworthiness that are often observed in stressed market environments (see also \cite{ly19} for a recent application). This direction aligns with empirical evidence showing that default clustering and tail risk are often underrepresented in standard Gaussian models. More recently, Tang et al. \cite{TTY2021} developed a dynamic framework to model the underlying drivers of default risk and established a large-portfolio loss limit, the characteristics of which were analyzed through Monte Carlo simulation.

Recent studies have also incorporated environmental risk dimensions into credit models. For example, in addition to the work already mentioned by Agliardi \cite{agl21}, Battiston et al. \cite{battiston2017} discussed the implications of climate-related risks on financial stability, suggesting the need to adapt traditional credit models to incorporate transition and physical risk channels. Javadi and Masum \cite{jm21} found empirical evidence that firms situated in regions more exposed to climate change tend to incur significantly higher interest rate spreads on their bank loans. Similarly, Huang et al. \cite{hkww22} showed that elevated climate risk is associated with less favorable bank loan terms. Furthermore, higher levels of climate risk correlate with weaker financial performance and increased default probability, potentially resulting in more stringent lending conditions.


In the present study, we contribute to this strand of literature by explicitly considering a credit portfolio partitioned into two subportfolios: one composed of loans classified as "green" and the other including all remaining loans, hereafter referred to as "brown" loans. Extending the classical Vasicek framework, we employ a factor copula approach to capture default dependencies among borrowers. The dependence structure is governed by two distinct systematic risk factors, while idiosyncratic risk components are modeled, as in the standard setting, using independent standard Gaussian random variables. By introducing additional distributional assumptions on the systematic factors, we allow for skew-normal specifications. This generalization introduces a skewness parameter, thereby extending the classical Gaussian framework and enabling the model to capture potential asymmetries in borrower (log-)creditworthiness.
Moreover, our model incorporates a non-uniform distribution of relative exposures across obligors. Under mild regularity conditions, we establish a quadratic mean convergence result for the partitioned portfolio loss, showing that it converges to a limiting random variable that reflects the distinct contributions of green and brown loan segments. The convergence condition can be interpreted in terms of the infinitesimal Herfindahl-Hirschman concentration index, linking portfolio granularity with asymptotic behavior of the index (see Gordy \cite{gordy} and Emmer and Tasche \cite{EmmTasc}). 
The limiting loss variable is fully characterized by the default probabilities and factor loadings associated with the two subportfolios, thereby allowing for differentiated default behavior and sensitivity to systematic risk across green and brown loan segments. Furthermore, we focus on the study of the convergence properties of the portfolio loss distribution, emphasizing the effects of exposure concentration and model parameters, such as default probabilities, factor loadings, and skewness, on value-at-risk.

The remainder of the paper is structured as follows. Section 2 introduces the modeling framework and presents the main theoretical results, with particular attention to the convergence assumptions, which are supported by empirical observations. Section 3 presents a set of numerical experiments aimed at assessing the behavior of losses in green-brown portfolios under the proposed credit risk model. Finally, Section 4 concludes with a summary of the key findings and directions for future research.


\section{Mixed green-brown loan portfolios and their LHP approximation}
The approximation introduced by Vasicek (\cite{vas1},\cite{vas2}) for computing losses in a credit portfolio of $N$ loans due to defaults of a fraction of the borrowers is based on the following set of assumptions:\\

\noindent
 {\bf{(V1)}} the number of borrowers is very large;\\
 {\bf{(V2)}} all borrowers share the same default probability;\\
 {\bf{(V3)}} logarithmic returns of borrowers are described by  the the following one-factor model
\begin{equation}\label{basic}
Y^{(N)}_h=\rho X+\sqrt{1-\rho^2}Z^{(N)}_h\;\;h=1,\ldots,N, \text{and}\;\;\rho\in[0,1);
\end{equation}
 {\bf{(V4)}} the systematic factor $X$ is $N(0,1)$ distributed, the  idiosyncratic shocks $(Z^{(N)}_h)$ are i.i.d  $N(0,1)$ r.v.'s indipendent from $X$.\\ 
 
\noindent 
The last two assumptions imply that the chosen copula model for logarithmic returns is
\begin{equation}\label{gaussian1}
{\bf{Y}}^{(N)}\sim N({\bf{0}},{\bf{C}})
\end{equation}
with $N\times N$ covariance matrix ${\bf{C}}=(C_{ij})$ such that $C_{ii}=1$ and $C_{ij}=\rho^2$ for $i\neq j$, that is the correlation of borrower returns does not change across borrower pairs.\\
 Because of these assumptions the procedure is known as Large Homogeneous Portfolio (LHP) single factor approximation. Vasicek hinged his method on the adoption of the structural approach to default as proposed by Merton in \cite{mert}. According to Merton  default of the borrower labeled by "{\it{i}}\;" can only happen if its return falls below a certain threshold value $K_i$ at the repaying time $T$. Therefore assuming there are $N$ borrowers from a specific bank, ${\bf{u}}^{(N)}=(u_1,\ldots,u_N)$ being their exposure percentages\footnote{By denoting with $E_i$ the absolute exposure to obligor $i$, the exposure percentage is given by $u_i = E_i/\sum_{i}^N E_i$. We assume that these exposures are known and deterministic. In practice, the situation can be more complex. In fact, much of bank lending takes the form of credit lines, which grant the borrower a degree of control over the actual exposure amount. However, under certain conditions (see Gordy \cite{gordy}), the quantity $E_i$ can be interpreted as the expected exposure in the event of the obligor’s default.}, and denoting by ${\bf Y}^{(N)}=(Y_1,\ldots,Y_N)$ the vector of their returns at time $T$ and by $(R_1,\ldots,R_N)$ the recovery rates, the percentage of loss on the bank credit portfolio will be given by:
$$
L({\bf Y}^{(N)},{\bf{u}}^{(N)})=\sum_{i=1}^Nu_i(1-R_i) 1_{(Y_i\leq K_i)}.
$$
Since the present paper, as in \cite{vas1} and \cite{vas2}, is not focusing on recovering aspects but on distributional properties of the loss, i.e. its quantiles, we shall assume $R_1=\ldots=R_n=0$ and  moreover adopt the shorthand notation $L^N({\bf{u}}):=L({\bf{Y}}^{(N)},{\bf{u}}^{(N)})$, hence
\begin{equation}\label{loss}
L^N({\bf{u}}) = \sum_{i=1}^N u_i \mathbb{I}_{(Y_i\leq K_i)} = \sum_{i=1}^N u_i \mathbb{I}_{D_i}.
\end{equation}
Under the previous set of assumptions (V1)...(V4) with the further condition $u_i=\frac1{N}$ for $i=1\ldots,N$, that is, the exposure weights are all equal, Vasicek derived the distribution of portfolio losses and computed its quantiles simply substituting in the obtained formula the estimates for the default probability and asset correlation.\\

We shall now extend the Vasicek model in a double direction: we shall considered loans of two different types, and moreover, we shall adopt a two-factor model to explain borrower logarithmic returns.

From now on, we suppose that the bank has classified the loans in the portfolio into two different classes, the class of "brown loans" and that of "green loans", and let $N_b$ and $N_g$ be the number of loans in each of these two classes, so that $N=N_b+N_g$. Correspondingly, the sum (\ref{loss}) splits  into two distinct contributions given by brown and green loans :
\begin{equation} \label{g_b_Loss}
L^N({\bf{u}})  = \sum_{i=1}^{N_b} u^b_i \mathbb{I}_{D^b_i}+\sum_{j=1}^{N_g} u^g_j \mathbb{I}_{D^g_j}\;,
\end{equation}
where the suffixes "b" and "g" have been added to exposures and to defaulting events. From now on, we suppose the exposure percentages have been reordered in such a way that the vector ${\bf{u}}^{(N)}$ is partitioned as ${\bf{u}}^{(N)}=({\bf{u}}^b,{\bf{u}}^g)\in \mathbb{R}^{N_b}\times\mathbb{R}^{N_g}$. Similarly, partitioning ${\bf{1}}^{(N)}
=(1,\ldots,1)\in \mathbb{R}^N$  in the two components ${\bf{1}}=({\bf{1}}^b,{\bf{1}}^g)$, we have that
$\omega_b={\bf{1}}^b\cdot {\bf{u}}^b$ and $\omega_g={\bf{1}}^g\cdot { \bf{u}}^g$, respectively, represent the percentages of total exposure of the bank to brown loans and to green loans, therefore $\omega_b+\omega_g=1$. For example, $\omega_b\approx 0.75$ and $\omega_g\approx 0.25$ appear to be, on average, the consistent values for the credit portfolios of a bank in the European bank sector (\cite{FC2023}). We now make a set of assumptions inspired by the Vasicek approach recalled at the beginning of the section and similar to those appearing in the paper \cite{agl21}, that is:\\

\noindent
 {\bf{(bg1)}} the numbers of brown and green borrowers are very large;\\
 {\bf{(bg2)}} all borrowers in the same class share the same default probability;\\
 {\bf{(bg3)}} logarithmic returns of borrowers within each class are described by  the following two-factor model
\begin{equation}\label{basic}
Y^a_h=\rho_a\sqrt{1-\delta_a^2}\;X_1 +\rho_a\delta_a\; X_2+\sqrt{1-\rho_a^2}\;Z^a_{h}\;\;h=1,\ldots,N_a,\;a=b,g,
\end{equation}
with $\rho_a \in[0,1)$ and $\delta_a \in (-1,1)$;\\
 {\bf{(bg4)}} the systematic factors are independent and distributed accordingly to  $X_1\sim N(0,1)$, $X_2\sim |W|$ with $W\sim N(0,1)$, the  idiosyncratic shocks $(Z^a_h)$ are i.i.d  $N(0,1)$ r.v.'s independent from $X_1$ and $X_2$.\\ 
 
 \noindent
The last two assumptions imply that the chosen copula model for logarithmic returns is
\begin{equation}\label{skew-gaussian}
{\bf{Y}}^{(N)}={\bf{Y}}^{(N_b,N_g)}\sim SN({\bf{0}},{\bf{b}}{\bf{b}}'+{\bf{D}}{\bf{\Sigma}}, \frac{{\bf{\Sigma}}^{-1}{\bf{b}}}{\sqrt{1-{\bf{b}}'{\bf{\Sigma}}^{-1}{\bf{b}}}})\footnote{the symbol $SN(\xi,\Omega,\alpha)$ denotes the multivariate skew-normal distribution of location vector $\xi$, covariance matrix $\Omega$ and shape vector $\alpha$, see \cite{azz1}.},
\end{equation}
with ${\bf{b}}=(\beta_b{\bf{1}}^b,\beta_g{\bf{1}}^g)'$, where $\beta_a\equiv \delta_a\rho_a$ for $a=b,g$, ${\bf{D}}=diag(\sqrt{1-\beta^2_b}{\bf{1}}^b,\sqrt{1-\beta^2_g}{\bf{1}}^g)$, and ${\bf{\Sigma}}$ is a partitioned covariance matrix which entries are specified in Appendix.1, containing a proof of (\ref{skew-gaussian}). Moreover, the marginal $Y^a_h$ is (univariate) skew-normal $SN(\alpha_a)$, being $\alpha_a = \rho_a \delta_a /(\sqrt{1-\rho_a^2 \delta_a^2}$, hence
\begin{equation} \label{snorm_cdf}
F_{Y^a_h}(x) = 2 \int_{-\infty}^x N(\alpha_a y) dN(y).
\end{equation}

\noindent
{\bf{Remark}}: From (\ref{skew-gaussian}) we see that if $\delta_b=\delta_g=0$ then  ${\bf{b}}={\bf{0}}$, moreover  ${\bf{D}}=diag ({\bf{1}}^b,{\bf{1}}^g)={\bf{I}}$ and  the law of ${\bf{Y}}^{(N)}$ reduces to $SN({\bf{0}},{\bf{\Sigma}},{\bf{0}})=N(\bf{0},{\bf{\Sigma}})$. Therefore, resulting also ${\bf{\Sigma}}={\bf{C}}$, this reproduces (\ref{gaussian1}). \\

In the present model, the probability of default of a borrower in class "$a$" is given by $p^{(a)}\equiv \mathbb{P}(Y^{a} \leq K^{(a)})$. Given a statistical estimation $\hat p^{(a)}$ of the class defaulting probability and inverting the previous relation, it produces an estimation $\hat K^{(a)}$ of the class defaulting level by means of the relation $\hat K^{(a)} = F^{-1}_{Y^{a}}(\hat p^{(a)})$. \\

Next we introduce  the "per-class" Herfindahl-Hirschman Indexes 
$$
HHI_{N_b}({\bf{u}}^b)=\sum^{N_b}_{i=1} (u_i^b)^2,\;\;\;HHI_{N_g}({\bf{u}}^g)=\sum^{N_g}_{j=1} (u_j^g)^2. 
$$
If one of the two indexes index goes to zero as the number of the borrowers in that class increases this means that there is no loan which is significantly larger than all the others in the considered class. The following results hold:\\

\begin{Proposition}\label{Convergence}
If $\E[\cov(\mathbb{I}_{D^a_i},\mathbb{I}_{D^c_j}|X_1, X_2)]=0$, for $i \neq j$, $a,c \in \{b,g \}$ and $HHI_{N_b}({\bf{u}}^b), HHI_{N_b}({\bf{u}}^g) \rightarrow 0$ as $N_g, N_b \to \infty$, then
$$
L^N(\mathbf{u}) - \E[L^N(\mathbf{u}) | X_1, X_2 ]\to 0
$$
in probability.
\end{Proposition}

\begin{proof} We rewrite the loss (\ref{g_b_Loss}) as $L^N(\mathbf{u}) = \sum_{i=1}^N u_i \mathbb{I}_{D_i}$, where $u_i = u_i^b$ for $i=1, \ldots, N_b$, and $u_i = u_{i-N_b}^g$ for $i=N_b+1, \ldots, N_b+N_g$ (similarly for the default events $D_i$). We have
$$
\E[(L^{N}-\E(L^N|X_1, X_2))^2]=\sum_{i=1}^N u_i^2 \E[(\mathbb{I}_{D_i}-\E[\mathbb{I}_{D_i}|X_1, X_2 ])^2]
$$
$$
+2 \sum_{i<j}^N u_i u_j \E[(\mathbb{I}_{D_i}-\E[\mathbb{I}_{D_i}|X_1, X_2 ]) (\mathbb{I}_{D_j} - \E[\mathbb{I}_{D_j}|X_1, X_2 ])]
$$
$$
=\sum_{i=1}^N u_i^2\E(var(\mathbb{I}_{D_i}|X_1, X_2 )) + 2 \sum_{i<j}^N u_i u_j  \E[\cov(\mathbb{I}_{D_i},\mathbb{I}_{D_j}|X_1, X_2)].
$$

Since $var(\mathbb{I}_{D_i}|X_1, X_2 ) = \E[\mathbb{I}_{D_i}|X_1, X_2 ] (1-\E[\mathbb{I}_{D_i}|X_1, X_2 ]) \leq 1/4$, from the hypothesis we
immediately get
$$
\E[(L^{N}-\E(L^N|X_1, X_2 ))^2] \leq \frac{1}{4}\sum^N_{i=1} u_i^2 = \frac{1}{4} \left(\sum^{N_b}_{i=1} (u_i^b)^2 + \sum^{N_b}_{i=1} (u_i^g)^2 \right) \to 0.
$$
\end{proof}

\begin{Proposition}\label{Limit-distrib}
 {\bf{(i)}} Suppose $u^b_i = \omega_b/N_b$  and $u^g_j = \omega_g/N_g$, for $i=1,\ldots,N_b$ and $j=1,\ldots,N_g$ , then for $N_b$ and $N_g $ sufficiently large
$$
P(L^N({\bf{u}})\leq \ell) \approx P(L^{mix}_{\omega_b,\omega_g}\leq \ell),\; \forall \ell\in (0,1),
$$
with $L^{mix}_{\omega_b,\omega_g}$ given by the weighted sum
\begin{equation} \label{asymp1}
L^{mix}_{\omega_b,\omega_g}=\omega_bL_b(X^b)+\omega_gL_g(X^g),
\end{equation}
where $L_b(\cdot),L_g(\cdot)$ are explicit functions from $\mathbb{R}$ to $[0,1]$ depending on the parameters of the class and $X^b$ , $X^g$ are skew-normal distributed.\\
 {\bf{(ii)}} If $\delta_b=\delta_g\equiv \delta$ then $X^b\sim X^g$ and we have 
\begin{equation} \label{asymp2}
 L^{mix}_{\omega_b,\omega_g}=v_{\omega_b,\omega_g}(X)
\end{equation}
 with $v_{\omega_b,\omega_g}(\cdot)\equiv \omega_bL_b(\cdot)+\omega_gL_g(\cdot)$ and $X\sim SN(0,1,\frac{\delta}{\sqrt{1-\delta^2}})$. Moreover, the  density of $ L^{mix}_{\omega_b,\omega_g}$ admits the semi-analytical representation given by
\begin{equation}\label{dens}
f^{mix}_{\omega_b,\omega_g}(\ell) =\frac{2 N(\gamma x^*(\ell)) N'(x^*(\ell))}{ \sum_{a = b,g}\left(\frac{\omega_a\rho_a}{\sqrt{1-\rho_a^2}}\right)N'\left(\frac{F_{Y^{(a)}}^{-1}(p^{(a)})-\rho_a x^*(\ell)}{\sqrt{1-\rho_a^2}} \right) },
\end{equation}
over the unit interval, with $\gamma \equiv \frac{\delta}{\sqrt{1-\delta^2}}$, where $x^*(\ell)$ is, for each fixed $\ell \in (0,1)$, the unique root of the equation:
$$
\omega_b\!\underbrace{N\left(\frac{F_{Y^{(b)}}^{-1}(p^{(b)}) - \rho_b x} {\sqrt{1-\rho_b^2}}\right)}_{decreasing \;in\; x}\!\!+\omega_g\!\underbrace{N\left(\frac{F_{Y^{(g)}}^{-1}(p^{(g)}) - \rho_g x} {\sqrt{1-\rho_g^2}}\right)}_{decreasing \;in\; x}\!=\ell.
$$
\end{Proposition}
\begin{proof}
 {\bf{(i):}}: By the previous result is sufficient to show that $L^{mix}_{\omega_b,\omega_g}=\E[L^N(\mathbf{(u^b,u^g)}) | X_1, X_2 ]$ for all $N$. We have
$$
\E[L^N(\mathbf{(u^b,u^g)}) | X_1, X_2 ]=\sum_{i=1}^{N_b}u^b_i\E(\mathbb{I}_{D^b_i}|X_1,X_2)+\sum_{j=1}^{N_g}u^g_j\E(\mathbb{I}_{D^g_i}|X_1,X_2)
$$
$$
=\frac{\omega_b}{N_b}\sum_{i=1}^{N_b}P(D^b_i|X_1,X_2)+\frac{\omega_g}{N_g}\sum_{j=1}^{N_g}P({D^g_j}|X_1,X_2)
$$
$$
=\frac{\omega_b}{N_b}\sum_{i=1}^{N_b}P(Y_i^b\leq K^b|X_1,X_2)+\frac{\omega_g}{N_g}\sum_{j=1}^{N_g}P(Y_j^g\leq K^g|X_1,X_2)
$$
$$
={\omega_b}P(Y_1^b\leq K^b|X_1,X_2)+{\omega_g}P(Y_1^g\leq K^g|X_1,X_2),
$$
where the last equality follows from intra-class homogeneity. Since $Z_1^b$ and $Z_1^g$ are independent from $X_1$ and $X_2$ it follows from ${\bf{(bg3)}}$ that
$$
P(Y_1^a\leq K^a|X_1,X_2)=P\biggl(Z_1^a\leq \frac{K^a-\rho_a(\sqrt{1-\delta^2_a}X_1+\delta_aX_2)}{\sqrt{1-\rho^2_a}}\biggr)=N(h^a(X^a)),
$$
for $a\in \{b,g\}$ and where $N(\cdot)$ denotes the distribution of a standard normal , $h^a(x)\equiv \frac{K^a-\rho_a x}{\sqrt{1-\rho^2_a}}$ and $X^a \equiv \sqrt{1-\delta^2_a} X_1+\delta_a X_2 \sim SN(0,1,\frac{\delta_a}{\sqrt{1-\delta_a^2}})$. Henceforth we have
$$
\E[L^N(\mathbf{(u^b,u^g)}) | X_1, X_2 ]={\omega_b}N(h^b(X^b))+{\omega_g}N(h^g(X^g)),
$$
and the proof ends by setting $L_b\equiv N\circ h^b$ and $L_g\equiv N\circ h^g$; \\
 {\bf{(ii):}} Consider the distribution of the r.v. $v_{\omega_b,\omega_g}(X)$, that is 
 $$
 F_v(\ell)=P(v_{\omega_b,\omega_g}(X)\leq \ell)
 $$
 and notice that the function $x\to v_{\omega_b,\omega_g}(x)$ from  $\mathbb{R}$ to $[0,1]$ is strictly decreasing, therefore for each $\ell\in[0,1]$ the equation $v_{\omega_b,\omega_g}(x)=\ell$ admits a unique solution $x^*(\ell)\equiv v_{\omega_b,\omega_g}^{-1}(\ell)$, henceforth $F_v(\ell)=P(X\geq x^*(\ell))=1-P(X\leq x^*(\ell))$. By differentiating this last expression we obtain 
 $$
 \frac{d F_v(\ell)}{d\ell}=-f_X(x^*(\ell))\bigg(\frac{dv_{\omega_b,\omega_g}^{-1}(\ell)}{d\ell}\bigg)=\frac{f_X(x^*(\ell))}{\sum_{a = b,g}\left(\frac{\omega_a\rho_a}{\sqrt{1-\rho_a^2}}\right)N'\left(\frac{F_{Y^{(a)}}^{-1}(p^{(a)})-\rho_a x^*(\ell)}{\sqrt{1-\rho_a^2}} \right) }\;,
 $$
where $f_X$ denotes the density of $X$. Since this density has the form of the numerator of formula (\ref{dens}) the proof is finished.
\end{proof}

{\Remark \label{rem1} Under the hypothesis of Proposition \ref{Limit-distrib}, \textbf{(ii)}, by denoting with $VaR_{\beta}(L^{mix}_{\omega_b,\omega_g})$ the value-at-risk at level $\beta$ of the r.v. $L^{mix}_{\omega_b,\omega_g}$, we immediately have 
$$
P(v_{\omega_b,\omega_g}(X)\leq VaR_{\beta}(L^{mix}_{\omega_b,\omega_g})) = \alpha \ \ \Leftrightarrow \ \ P(X \leq v_{\omega_b,\omega_g}^{-1}(VaR_{\beta}(L^{mix}_{\omega_b,\omega_g}))) = 1-\beta
$$
being $v = \omega_b N\circ h^b + \omega_g N\circ h^g$ stricly increasing, from which we get $VaR_{1-\beta}(X) = v_{\omega_b,\omega_g}^{-1}(VaR_{\beta}(L^{mix}_{\omega_b,\omega_g}))$, finally giving
$$
VaR_{\beta}(L^{mix}_{\omega_b,\omega_g}) = v_{\omega_b,\omega_g}(VaR_{1-\beta}(X)) = \omega_b L_b(VaR_{1-\beta}(X))+\omega_gL_g(VaR_{1-\beta}(X)),
$$
where $L_b\equiv N\circ h^b$ and $L_g\equiv N\circ h^g$. The value-at-risk is therefore \textit{linear} w.r.t. the portfolio proportions $\omega_b, \omega_g$.

\medskip
 
Moreover, it is worth noting that if $\delta \equiv 0$, the systematic factor $X$ becomes standard normal $N(0,1)$, and the asymptotic loss rate (\ref{asymp2}) reduces to the model introduced in \cite{agl21}.

On the other hand, if $p^{(g)} = p^{(b)}$ and $\rho_g = \rho_b$, then $L_b = L_g \equiv L$. Hence, we get a non-partitioned credit portfolio that reproduces the classical Vasicek Gaussian LHP approximation model for $\delta = 0$. If $\delta \neq 0$, we end up with the extension of the LHP approximation to the case of skew-normal one-factor model (see \cite{lp15}),
\begin{equation} \label{asympSN}
    L = v_{\omega_b,\omega_g}(X) \equiv v(X), \ \ \ X \sim SN(0,1,\frac{\delta}{\sqrt{1-\delta^2}}).
\end{equation}

It is easily seen that, in such a case, the level-$\beta$ value-at-risk of $L$, is given by
$$
VaR_{\beta}(L) = N\left(\frac{K-\rho VaR_{1-\beta}(X)}{\sqrt{1-\rho^2}} \right),
$$
while the corresponding density is
$$
f_L(\ell) = 2 \frac{\sqrt{1-\rho^2}}{\rho} N(\alpha v^{-1}(\ell)) \exp\{-\frac{1}{2}(v^{-1}(\ell)^2-N^{-1}(\ell)^2 )\}
$$
where $v^{-1}(\ell) = (K-\sqrt{1-\rho^2}N^{-1}(\ell))/\rho$.
}

\bigskip

Loan portfolios are typically non-granular, with (normalized) exposure weights exhibiting diverse structural patterns, see \cite{Lutke08}. In light of our convergence result (see Proposition~\ref{Convergence}), we aim to identify sufficient conditions under which the Herfindahl–Hirschman Index (HHI) becomes asymptotically negligible as $N \longrightarrow +\infty$, even for broad classes of weight distributions. Specifically, we examine weights that follow a power-law decay, i.e. $u_i \propto i^{-a}$. In this setting, the corresponding concentration index takes generally the form
$$
HHI_N(\mathbf{u}) \equiv HHI_N(a)  = \frac{1}{C_N^2} \sum_{i=1}^N \frac{1}{i^{2a}}, 
$$
where $C_N \equiv \sum_{i=1}^N \frac{1}{i^a}$.

In order to find sufficient conditions for the convergence, we use the following (very general) asymptotic approximation, $b\neq 1$:
\begin{equation} \label{asyInteg}
\sum_{i=1}^N i^{-b} \approx \int_1^N x^{-b} dx = \left\{ \begin{array}{ll} \frac{N^{1-b}-1}{1-b} & b < 1 \\  \\ \frac{1- N^{1-b}}{b-1}  & b>1. \end{array} \right.
\end{equation}

Now, by using the approximation (\ref{asyInteg}), we get
$$
HHI_N(a) \approx \begin{cases} \frac{(1-a)^2}{1-2a} \frac{N^{1-2a} - 1}{(N^{1-a}-1)^2} & a <1/2 \\  \\
\frac{(1-a)^2}{2a-1}  \frac{1-N^{1-2a} }{(N^{1-a}-1)^2} &  1/2 < a < 1 \\ \\ 
\frac{(a-1)^2}{2a-1}  \frac{1-N^{1-2a} }{(1-N^{1-a})^2} &  a>1 
\end{cases}   = \begin{cases} O(N^{-1})        & a <1/2 \\  \\
                              O(N^{-2(1-a)}) &  1/2 < a < 1 \\ \\ 
                              O(1)                 &  a>1 
                \end{cases}                 
$$
so that we may conclude
\begin{equation}
\lim_{N\rightarrow + \infty} HHI_N(a) = \begin{cases} 0 & a <1 \\ \\  \frac{(a-1)^2}{2a-1}  & a > 1. \end{cases}
\end{equation}

This choice in the characterization of percentage exposures is motivated by empirical evidence drawn from the analysis of exposure data for a selected group of systemically important banks: in particular, we considered BNP Paribas (BNP), Mitsubishi UFJ Financial Group (MUFG), Crédit Agricole (CA), Société Générale (SG), HSBC, and Mizuho Financial Group (MHFG). Figure \ref{Fig:exposures} presents the distribution of loan exposures, sorted in descending order. The analysis is based on publicly available deal-level data sourced from Refinitiv (now LGSE) and should be interpreted as a proxy for the true distribution of credit exposures. This approximation is necessary due to incomplete disclosure of commitment shares and role-specific allocations across syndicated loans.
For each selected bank, various loan instruments are included—most notably term loans, revolving credit facilities, and standby letters of credit—where the bank may act as Lead Arranger, Mandated Arranger, Arranger, Lender/Participant, Book Runner, Admin Agent, Syndication Agent..... Since commitment shares were not reported, we adopted a uniform allocation assumption by distributing the loan amount equally among all participating institutions, consistent with practices in the absence of granular data. The final exposure figures were then aggregated at the borrower level. While this introduces simplification, the approach remains informative for comparative exposure analysis and aligns with methodologies commonly employed in empirical studies on credit networks and systemic risk (cf. Cont et al., 2013; Glasserman and Young, 2016).
In practice, we fitted the model $u_i = c/(b+i)^a$ to the observe ordered percentage exposures of each bank. The results of the fitting exercise are reported in Table \ref{powlaw_results}. Furthermore, by cross-referencing the loan-level data obtained from Refinitiv with the environmental classification information available in the EF dataset for the selected banks, we estimate that green loans account, on average, for approximately $25\%$ of total loan exposures, with observed values ranging from a minimum of $8\%$ to a maximum of $40\%$.

\begin{figure}[h]
\begin{center}
\includegraphics[width=18cm,height=12cm]{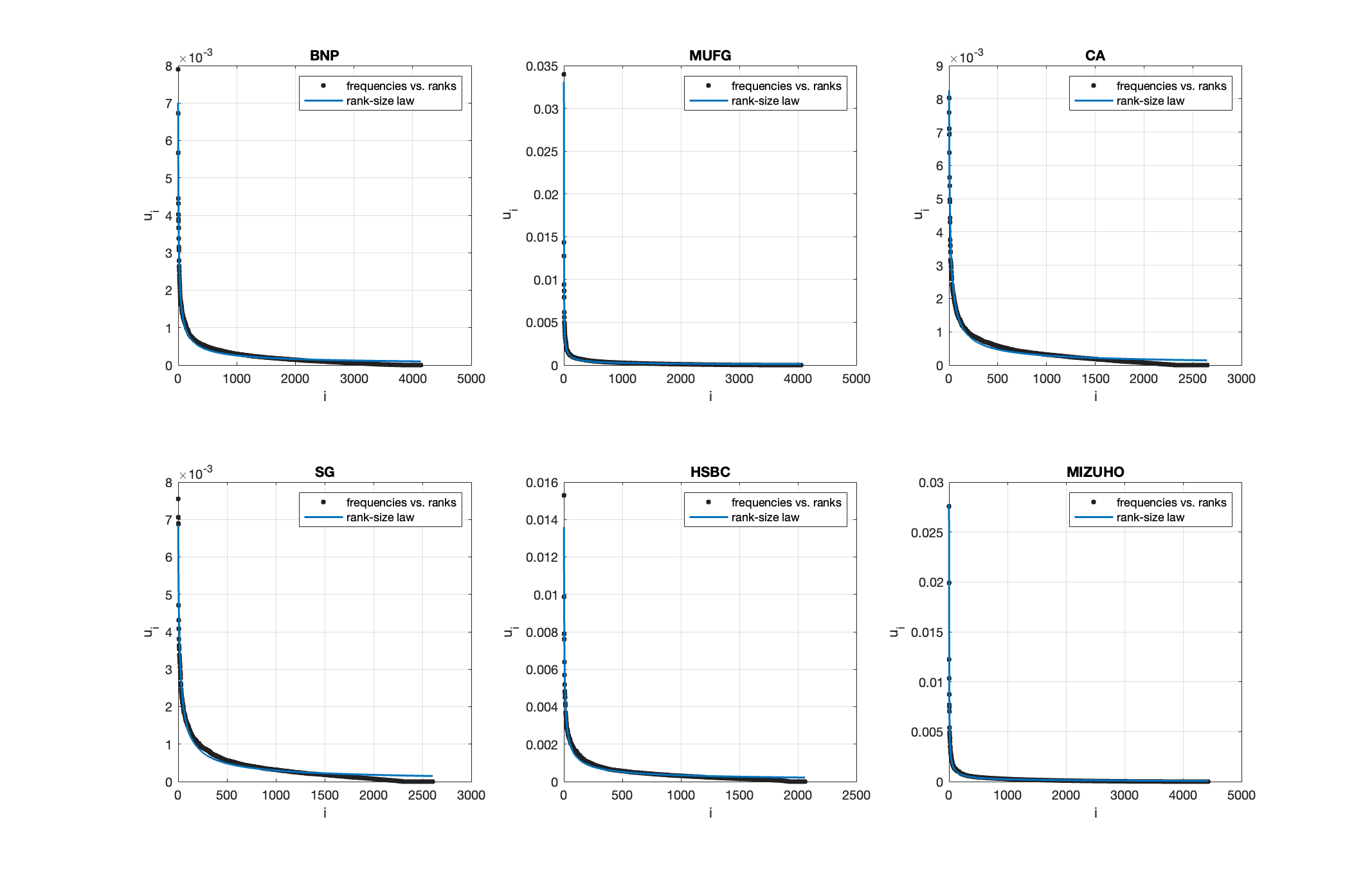}
\end{center}
\caption{Observed percentage exposures of main banks (dots) and fitted power-law model: data from Refinitiv - Loan App. }
\label{Fig:exposures}
\end{figure}

\begin{table}[h!]
\centering
\begin{tabular}{l
                S[table-format=1.4]
                S[table-format=1.4]
                S[table-format=1.4]
                S[table-format=1.3]}
\toprule
\textbf{Bank} & \textbf{a} & \textbf{b} & \textbf{c} & \textbf{HHI} \\
\midrule
\multirow{3}{*}{BNP}
& 0.6888 & 3.399 & 0.0308 & {---} \\
& {(0.6835, 0.6942)} & {(3.2190, 3.5800)} & {(0.0299, 0.0318)} & \num{8.92e-04} \\
& \multicolumn{2}{l}{rmse = \num{7.08e-05}} & \multicolumn{2}{l}{r2 = \num{0.9681}} \\
\midrule
\multirow{3}{*}{MUFG}
& 0.6359 & -0.378 & 0.02059 & {---} \\
& {(0.6319, 0.6399)} & {(-0.3905, -0.3655)} & {(0.02023, 0.02096)} & \num{2.68e-03} \\
& \multicolumn{2}{l}{rmse = \num{1.05e-04}} & \multicolumn{2}{l} {r2 = \num{0.9818}} \\
\midrule
\multirow{3}{*}{CA}
& 0.738 & 4.859 & 0.04838 & {---} \\
& {(0.7308, 0.7453)} & {(4.554, 5.163)} & {(0.04641, 0.05036)} & \num{1.29e-03} \\
& \multicolumn{2}{l}{rmse = \num{1.0025e-04}} & \multicolumn{2}{l}{r2 = \num{0.9708}} \\
\midrule
\multirow{3}{*}{SG}
& 0.7446 &  6.714 & 0.05286 & {---} \\
& {(0.7357, 0.7535)} & {(6.209, 7.219)} & {(0.05012, 0.05559)} & \num{1.1987e-03} \\
& \multicolumn{2}{l}{rmse = \num{1.1045e-04}} & \multicolumn{2}{l}{r2 = \num{0.9610}} \\
\midrule
\multirow{3}{*}{HSBC}
& 0.626 & 0.894 & 0.02575 & {---} \\
& {(0.6186, 0.6333)} & {(0.8, 0.988)} & {(0.02474, 0.02675)} & \num{1.6165e-03} \\
& \multicolumn{2}{l}{rmse = \num{1.3622e-04}} & \multicolumn{2}{l}{r2 = \num{0.9662}} \\
\midrule
\multirow{3}{*}{MHFG}
& 0.7084 & 0.1218 & 0.03117 & {---} \\
& {(0.7048, 0.7119)} & {(0.1012, 0.1423)} & {(0.03069, 0.03165)} & \num{2.5639e-03} \\
& \multicolumn{2}{l}{rmse = \num{8.6795e-05}} & \multicolumn{2}{l}{r2 = \num{0.9857}} \\
\bottomrule
\end{tabular}
\caption{Power-law fitting results for the percentage exposures, $u_i = c/(b+i)^a$. Root mean squared errors (rmse) and R-squared goodness-of-fit (r2), as obtained by the MatLab Non-linear least-square function are reported.} \label{powlaw_results}
\end{table}

\section{Numerical Results}
In this section, we present numerical experiments designed to investigate the behavior of green-brown portfolio losses under the proposed credit risk model. We first describe the simulation framework, which combines a power-law distribution of exposures with a skewed systematic risk factor that affects the two distinct portfolio segments. The convergence of the simulated loss distribution towards its asymptotic limit is examined, highlighting the impact of exposure granularity and model parameters on the value-at-risk. Subsequently, we performed a sensitivity analysis to assess how variations in default probabilities, asset correlations, and exposure concentration affect the shape and dispersion of the limit distribution. These results provide insights into the robustness of the asymptotic approximation and the influence of portfolio composition, in particular, on the value-at-risk.

\subsection{Convergence analysis}

The simulation of portfolio losses is highly influenced by the probabilities of default (PDs), which determine the corresponding default thresholds. In analyzing the properties of the Large Homogeneous Portfolio (LHP) approximation, we focus on two baseline scenarios characterized, respectively, by $ p_G < p_B$ and $ p_G > p_B$, under the assumption of identical factor loadings.
In the first scenario, we set values $p_G = 0.005$ and $p_B = 0.01$, which are representative of typical PDs for large corporate exposures. The factor loading was fixed at $\rho = 0.1$, consistent with normal market conditions for corporate portfolios. 
In the second scenario, we assign $p_G = 0.2$ and $p_B = 0.15$, together with a moderately higher asset correlation of $\rho = 0.15$. This setting reflects a portfolio composed predominantly of green loans extended to medium-sized enterprises (SMEs), which are typically associated with higher default risk and slightly elevated asset correlations (see, e.g., \cite{EBA2024}, \cite{EBACasel2023}).
For each scenario, we further investigate the impact of portfolio concentration considering two different decay parameters $a$ for the exposure distribution. Specifically, we assume a power-law distribution for the ranked exposures $u_i \propto 1/i^{a}, \, i = 1, \ldots, n $. In the first case, we set $a = 0$, corresponding to a uniform (infinitely granular) portfolio. In the second case, we use $a = 0.6$, a value empirically observed across a sample of major banks (see Table~\ref{powlaw_results}).
This modeling framework allows us to capture both the effects of PD heterogeneity and exposure concentration on the performance of the LHP approximation under different market conditions.

We simulated the loss of portfolios of increasing size, from $n=500$ up to $n=5000$, according to the model described in Sect. 1, i.e., credit defaults are simulated following a factor structural model with systematic skewness. For each simulation, a common skew-normal factor $X\sim SN(\alpha)$ is generated to represent systematic risk, while idiosyncratic risk terms $Z_i \sim N(0,1)$ are drawn independently for each obligor. Defaults occur when the asset value falls below a pre-computed threshold $K$, calibrated to match the target probability of default, in our case $p_G$ and $p_B$, and are obtained by inverting the skew-normal cdf (see (\ref{snorm_cdf})). We fix the skewness parameter $\alpha = 0.5$. The portfolio loss in each simulation is calculated as the weighted sum of individual default indicators, taking into account the normalized exposures.
A total of $N_{samp}=10^{6}$  Monte Carlo replications are performed to estimate the portfolio loss distribution, from which empirical quantiles are calculated for three different levels, $\beta = 0.99, 0.995, 0.999$. The asymptotic theoretical quantiles are finally obtained by numerically inverting the cdf of the limiting LHP approximation. 
The results of the numerical simulation are reported in Figure \ref{fig:varconvergence}, where the absolute error is plotted against the dimension of the portfolio, and Tables \ref{Tab:convscenario1}-\ref{Tab:convscenario2}, in which the empirical quantiles and the error with respect to the asymptotic theoretical quantiles are shown with the corresponding value of the concentration of the green and brown sub-portfolios, $HHI_g$ and $HHI_b$, respectively. 

The numerical experiments confirm the theoretical expectations regarding the convergence properties of the simulated portfolio losses toward their asymptotic limit. First, convergence deteriorates as the exposure concentration increases, corresponding to higher values of the power-law decay parameter $a$. This behavior is consistent with the convergence theorem for large portfolios, which requires a sufficient degree of granularity to ensure accurate asymptotic approximations. Second, convergence becomes slower in the extreme tail of the loss distribution, particularly at high quantile levels such as the $99.9\%$ percentile. Moreover, convergence is worse under Scenario 1, where default probabilities are lower, leading to sparser default events. In such settings, rare default occurrences seem to delay the stabilization of the distribution toward its theoretical limit. Quantitatively, the relative error at the $\beta = 99.9\%$ quantile increases significantly with exposure concentration: for Scenario 1, it rises from approximately $7\%$ for the granular case $(a=0)$ to about $44\%$ when $a=0.6$. In Scenario 2, where the default rates are higher, the corresponding error increases more moderately, from roughly $2.3\%$ to $10\%$. These findings emphasize the crucial role of both portfolio structure and credit quality in determining the accuracy of asymptotic risk estimates, particularly in the tails of the distribution.

\begin{figure}[h]
\begin{center}
\includegraphics[width=6.5cm,height=5cm]{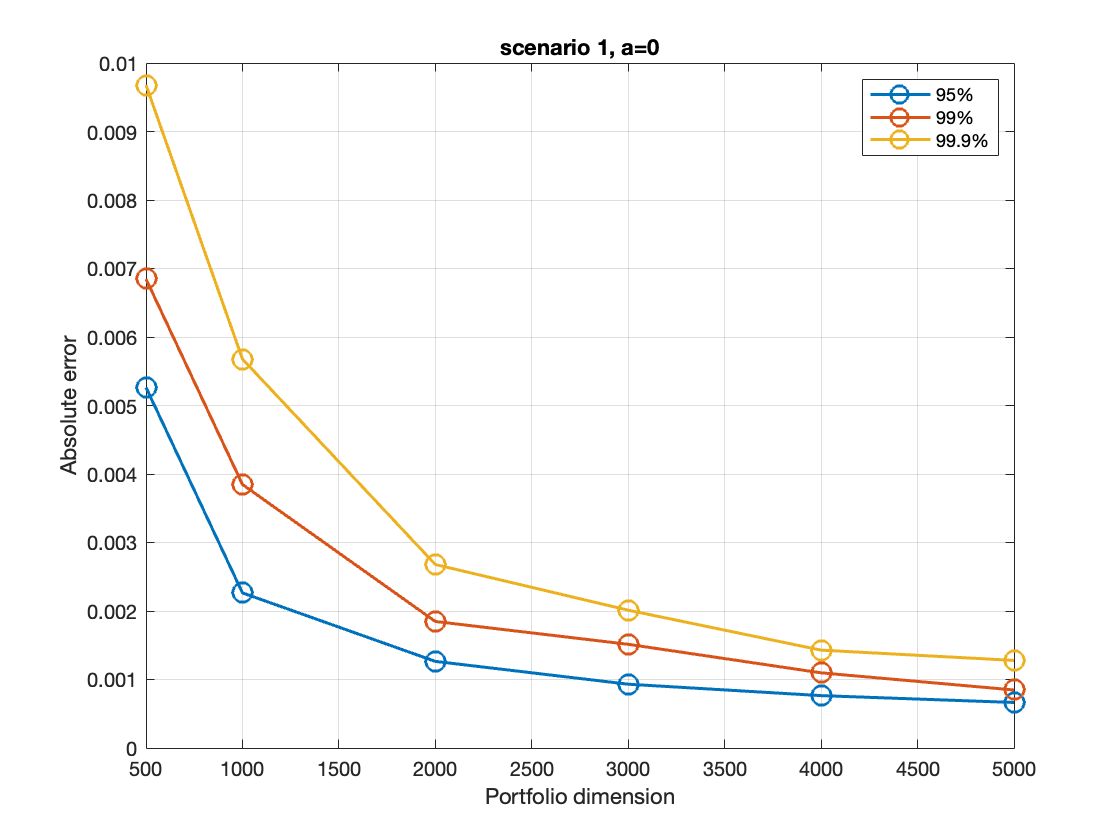}
\includegraphics[width=6.5cm,height=5cm]{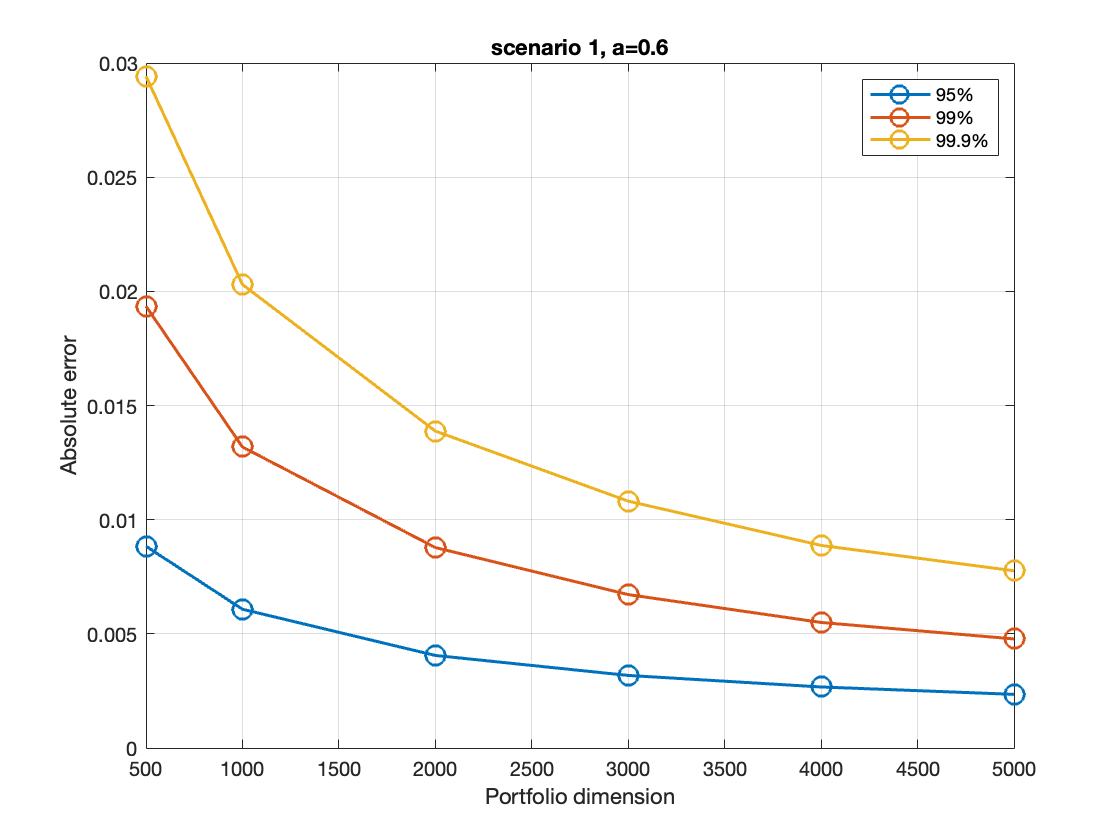}
\includegraphics[width=6.5cm,height=5cm]{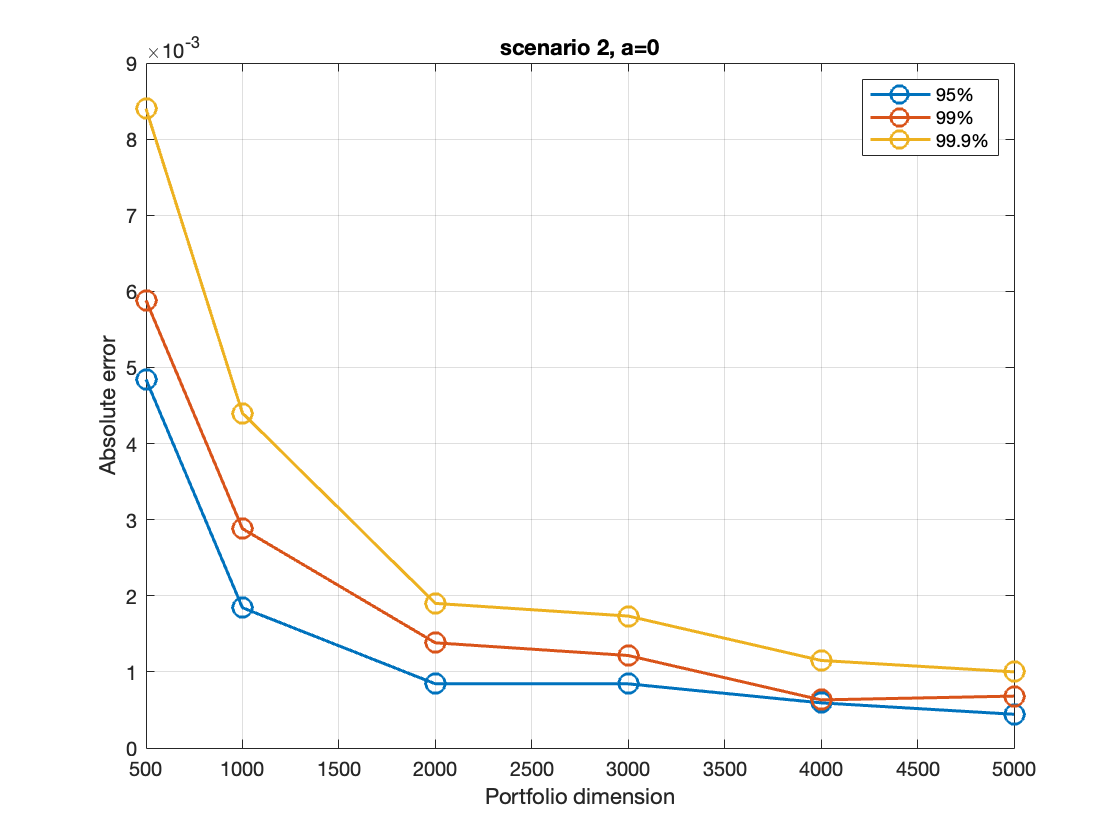}
\includegraphics[width=6.5cm,height=5cm]{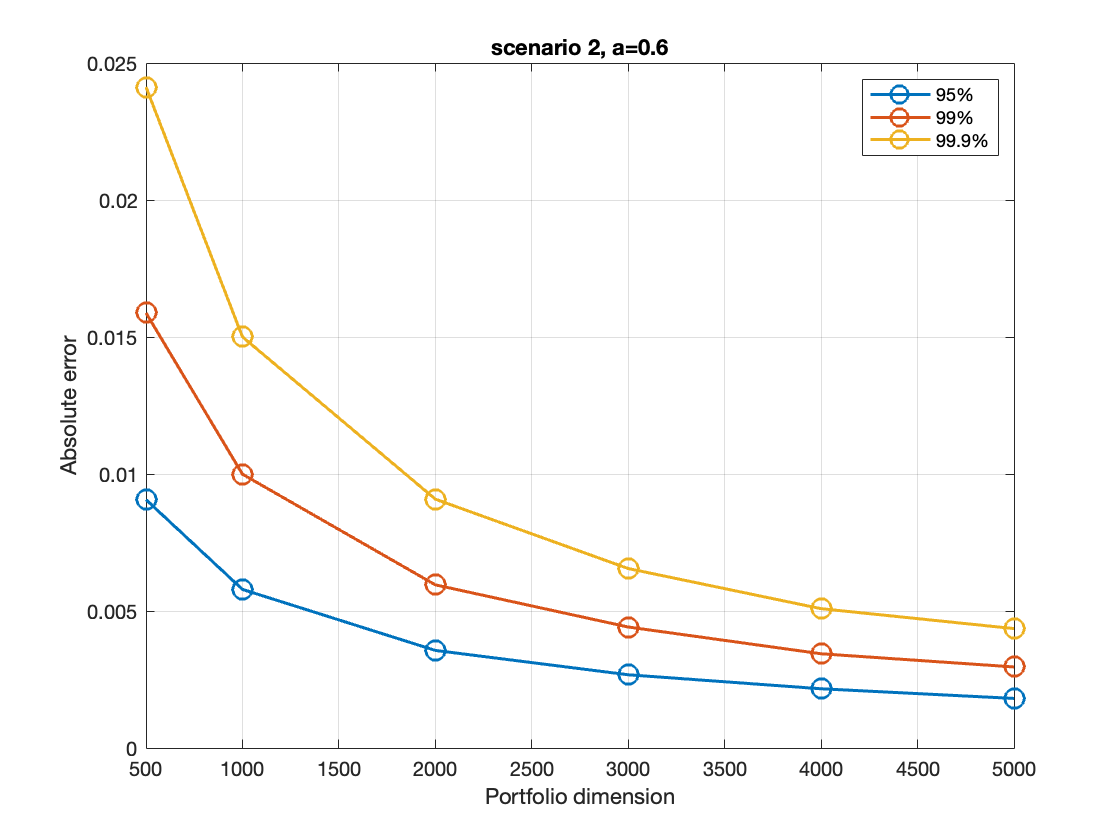}
\end{center}
\caption{Absolute errors between the  empirical quantiles and the theoretical quantile for increasing portfolio dimension for the four considered scenarios.}
\label{fig:varconvergence}
\end{figure}

\begin{table}[h!]
\centering
\begin{tabular}{l|ccccc|ccccc|}
\hline \hline
 & \multicolumn{3}{c}{\textbf{VaR}} & $\mathbf{HHI}_g$ &$\mathbf{HHI}_b$  & &  {\textbf{VaR}} & &$\mathbf{HHI}_g$ &$\mathbf{HHI}_b$\\
\cmidrule(lr){2-4} \cmidrule(lr){5-6} \cmidrule(lr){7-9} \cmidrule(lr){10-11}
 $n$ & 99\% & 99.5\% & 99.9\% &  \multicolumn{2}{c|}{$a=0$}  & 99\% &99.5\% & 99.9\% & \multicolumn{2}{c|}{$a=0.6$} \\
\midrule

\multirow{2}{*}{$500$} 
   & 0.0180 & 0.0220 & 0.0280 & & &  0.0212  &  0.0339 &   0.0468 & &\\
   & (0.0053) &   (0.0068) &   (0.0097) & 0.0066 & 0.0028 & (0.0088) &   (0.0193)  &  (0.0294) & 0.0136 & 0.0070\\
\midrule

\multirow{2}{*}{$1000$} 
   & 0.0150  &  0.0190  &  0.0240 & & & 0.0185  &  0.0278  &  0.0377 & &\\
   & (0.0023)  &  (0.0038)  &  (0.0057) & 0.0033 & 0.0014 & (0.0061) &   (0.0132) &   (0.0203) & 0.0079 & 0.0040\\
\midrule

\multirow{2}{*}{$2000$} 
   & 0.0140  &  0.0170  &  0.0210 & & & 0.0165  &  0.0234  &  0.0313 & &\\
   &(0.0013)  &  (0.0018)  &  (0.0027) & 0.0017 & 0.0007 & (0.0041)  &  (0.0088)  &  (0.0139) & 0.0046 & 0.0023\\
\midrule

\multirow{2}{*}{$3000$} 
   & 0.0137  &  0.0167  &  0.0203 & & & 0.0156  &  0.0213  &  0.0282& &\\
   & (0.0009)  &  (0.0015)  &  (0.0020) & 0.0011 & 0.0005 & (0.0032)  &  (0.0067)  &  (0.0108) & 0.0033 & 0.0017 \\
\midrule

\multirow{2}{*}{$4000$} 
   & 0.0135  &  0.0162  &  0.0198 & & & 0.0151  &  0.0201  &  0.0262& &\\
   & (0.0008)  &  (0.0011)  &  (0.0014) & 0.0008 & 0.0004 & (0.0027)  & (0.0055) & (0.0089) & 0.0026 & 0.0014\\
\midrule

\multirow{2}{*}{$5000$} 
   & 0.0134 & 0.0160 & 0.0196 & & & 0.0148 & 0.0194 & 0.0251 & &\\
   & (0.0007) &   (0.0008)  &  (0.0013) & 0.0007 & 0.0003 &  (0.0023)  &  (0.0048)  &  (0.0078) & 0.0022 & 0.0011\\
\hline \hline
\end{tabular}
\caption{Numerical results related to converge analysis for scenario 1, with $a=0$ and $a=0.6$. Empirical quantiles are reported, and, in parentheses, the error with respect to the theoretical value, $\widehat{VaR_{\beta}} - VaR_{\beta}$.  }
\label{Tab:convscenario1}
\end{table}

\begin{table}[h!]
\centering
\begin{tabular}{l|ccccc|ccccc|}
\hline \hline
 & \multicolumn{3}{c}{\textbf{VaR}} & $\mathbf{HHI}_g$ &$\mathbf{HHI}_b$  & &  {\textbf{VaR}} & &$\mathbf{HHI}_g$ &$\mathbf{HHI}_b$\\
\cmidrule(lr){2-4} \cmidrule(lr){5-6} \cmidrule(lr){7-9} \cmidrule(lr){10-11}
 $n$ & 99\% & 99.5\% & 99.9\% &  \multicolumn{2}{c|}{$a=0$}  & 99\% &99.5\% & 99.9\% & \multicolumn{2}{c|}{$a=0.6$} \\
\midrule

\multirow{2}{*}{$500$} 
   &  0.0320  &  0.0400  &  0.0520 & & &  0.0354  &  0.0484 &   0.0649 & &\\
   & (0.0048) &   (0.0059) &   (0.0084) & 0.0066 & 0.0028 & (0.0091) &   (0.0159)  &  (0.0241) & 0.0136 & 0.0070\\
\midrule

\multirow{2}{*}{$1000$} 
   & 0.0290  &  0.0370  &  0.0480 & & & 0.0321  &  0.0425  &  0.0558 & &\\
   & (0.0018)  &  (0.0029)  &  (0.0044) & 0.0033 & 0.0014 & (0.0058) &   (0.0100) &   (0.0150) & 0.0079 & 0.0040\\
\midrule

\multirow{2}{*}{$2000$} 
   & 0.0280  &  0.0355  &  0.0455 & & & 0.0299  &  0.0384  &  0.0499 & &\\
   &(0.0008)  &  (0.0014)  &  (0.0019) & 0.0017 & 0.0007 & (0.0036)  &  (0.0060)  &  (0.0091) & 0.0046 & 0.0023\\
\midrule

\multirow{2}{*}{$3000$} 
   & 0.0280  &  0.0353  &  0.0453 & & & 0.0290  &  0.0369  &  0.0473& &\\
   & (0.0008)  &  (0.0012)  &  (0.0017) & 0.0011 & 0.0005 & (0.0027)  &  (0.0044)  &  (0.0065) & 0.0033 & 0.0017 \\
\midrule

\multirow{2}{*}{$4000$} 
   & 0.0277  &  0.0348  &  0.0447 & & & 0.0285  &  0.0359  &  0.0459& &\\
   & (0.0006)  &  (0.0006)  &  (0.0012) & 0.0008 & 0.0004 & (0.0022)  & (0.0034) & (0.0051) & 0.0026 & 0.0014\\
\midrule

\multirow{2}{*}{$5000$} 
   & 0.0276  &  0.0348  &  0.0446 & & & 0.0281 & 0.0355 & 0.0452 & &\\
   & (0.0004) &   (0.0007)  &  (0.0010) & 0.0007 & 0.0003 &  (0.0018)  &  (0.0030)  &  (0.0044) & 0.0022 & 0.0011\\
\hline \hline
\end{tabular}
\caption{Numerical results related to converge analysis for scenario 2, with $a=0$ and $a=0.6$. Empirical quantiles are reported, and, in parentheses, the error with respect to the theoretical value, $\widehat{VaR_{\beta}} - VaR_{\beta}$.  }
\label{Tab:convscenario2}
\end{table}

\subsection{Sensitivity Analysis of Model Parameters}

In the second part of the numerical analysis, we examine how the shape and risk profile of the limiting loss distribution respond to key model parameters. As a preliminary step, we consider an unpartitioned portfolio to isolate the impact of the skewness parameter of the systematic risk factor on tail risk, as measured by the value-at-risk. Subsequently, we extended the analysis to the two-segment (green–brown) portfolio structure, exploring the effects of varying default probabilities and factor loadings. This allows us to assess how heterogeneity in credit quality and exposure to systematic risk influence the distributional characteristics of portfolio losses under the large portfolio approximation.

\paragraph{The role of $\alpha$.} 
To isolate the effect of the skewness parameter on portfolio risk, we first analyzed an unpartitioned portfolio under four representative configurations of market conditions. For normal market conditions, we considered two distinct borrower segments: large corporations and small-to-medium enterprises (SMEs). For these, we selected typical combinations of default probability and factor loadings, $(p, \rho)$, that is, $(0.01, 0.15)$ for large firms and $(0.03, 0.20)$ for SMEs. To model stressed market environments, we increased both parameters and considered $(0.04, 0.20)$ and $(0.06, 0.25)$, respectively. For each configuration, we varied the skewness parameter $\alpha$ of the systematic risk factor over a range of positive and negative values to assess its effect on tail risk.

The results shown in Figure \ref{fig:var_alpha} confirm the intuitive pattern that the value-at-risk increases with worsening credit quality, that is, with higher PD and correlation values. However, a less expected outcome emerged regarding the role of skewness: across all market scenarios, introducing skewness, either positive or negative, led to a decrease in the estimated value-at-risk at high confidence levels with respect to the normal ($\alpha\equiv 0$) model. This behavior suggests that, under the model assumptions, the presence of skew in the systematic factor introduces asymmetries in the loss distribution that reduce the likelihood of extreme losses in the upper tail. 

\begin{figure}[h]
\centering
\includegraphics[width=16cm,height=8cm]{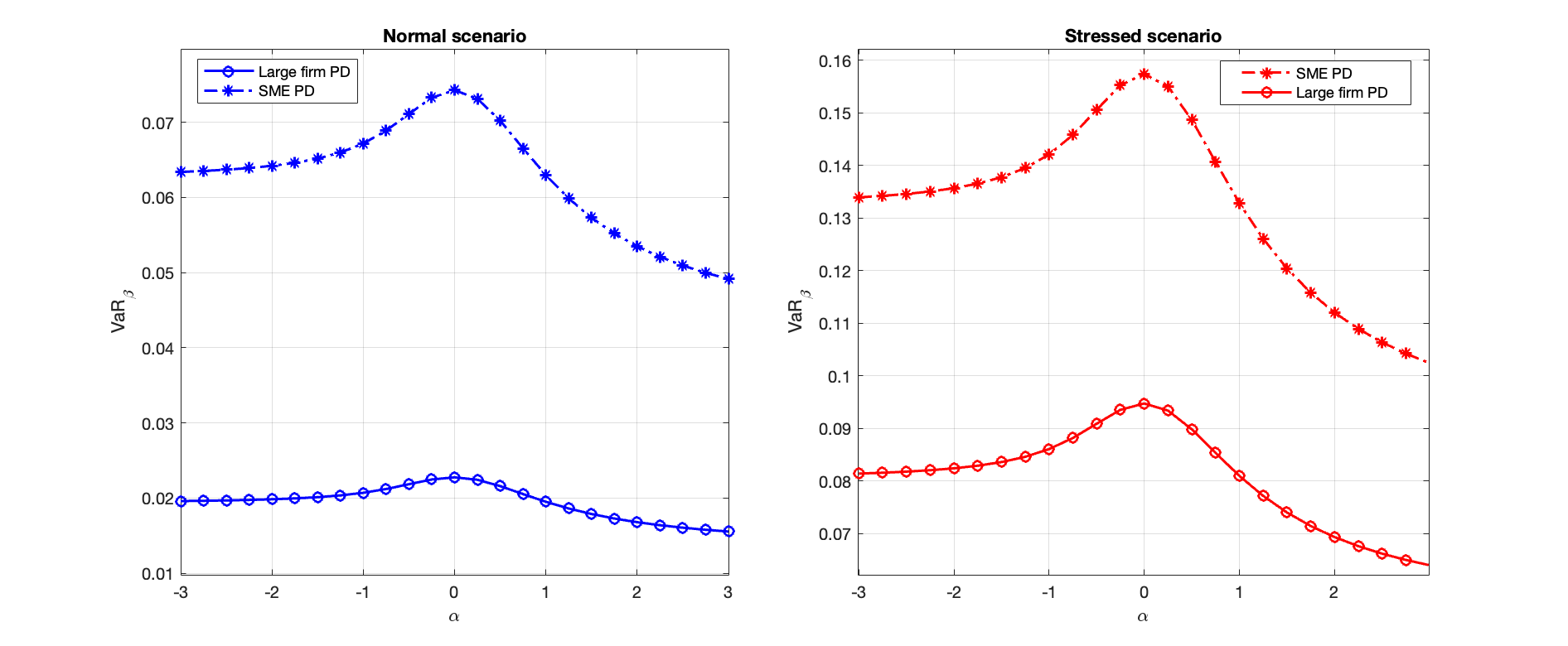}
\caption{The value-at-risk at level $\beta = 0.99$ for different values of the parameter $\alpha$. The four curves correspond to typical values for $(p,\rho)$ related to normal and stressed market conditions.}
\label{fig:var_alpha}
\end{figure} 

\paragraph{The green-brown portfolio.}

To investigate the influence of model parameters on the LHP approximation of the loss distribution for the green-brown portfolios, four distinct scenarios were designed to isolate the effects of differences in default probabilities and factor loadings between the two segments:

\begin{enumerate}
    \item \textit{Scenario 1}: $ \rho_G < \rho_B $, $ p_G = p_B $. Both sub-portfolios share the same credit quality, but green loans are assumed to be less sensitive to systematic risk than brown loans. This setup models a situation in which green lending is more sensitive to idiosyncratic shocks than to systemic changes.
    
    \item \textit{Scenario 2}: $ \rho_G > \rho_B $, $p_G = p_B $. Again, default probabilities are equal, but here green loans exhibit higher exposure to the systematic factor. This might represent cases where green investments are concentrated in sectors more vulnerable to macroeconomic shocks (e.g., renewable energy under volatile subsidies or market prices).
    
    \item \textit{Scenario 3}: $ \rho_G = \rho_B $, $p_G < p_B $. Both segments respond equally to systematic risk, but green loans are characterized by lower probabilities of default, reflecting a higher average credit quality. This configuration corresponds to the view that green investments tend to attract more stable borrowers or benefit from external support mechanisms.
    
    \item \textit{Scenario 4}: $ \rho_G = \rho_B $, $ p_G > p_B $. The reverse of Scenario 3, this setup assumes green loans are riskier in terms of default likelihood, despite similar exposure to systematic shocks. This might reflect early-stage or innovative projects with elevated credit risk profiles.
\end{enumerate}
These configurations allow us to clarify how the asymmetries in risk characteristics between green and brown exposures affect the shape and behavior of the limiting loss distribution.

The parameter values for the sub-portfolios were selected based on a baseline default probability of $p_G = p_B = 0.028$, with asset correlations set to $\rho_G = 0.10$ and $\rho_B = 0.15$ in Scenario 1, and $\rho_G = 0.15$ and $\rho_B = 0.10$ in Scenario 2. For Scenarios 3 and 4, which introduce asymmetry in default probabilities, we set $p_G = 0.02$, $p_B = 0.028$ (Scenario 3), and $p_G = 0.030$, $p_B = 0.028$ (Scenario 4), while maintaining a common asset correlation $\rho_G = \rho_B = 0.10$. Across all configurations, the portfolio composition was kept fixed, with green loans accounting for 25\% of total exposure, i.e., $\omega_G = 0.25 = 1 - \omega_B$, as suggested by empirical findings (see Sect. 2).

For each scenario, we report the density of $L^{mix}_{\omega_b,\omega_g}$ given by  (\ref{dens}), see Figure \ref{fig_densities}. It is evident that the right tail of the distribution is affected by a non-zero value of the parameter $\alpha$, and that a positive $\alpha$ has a slightly greater impact. This observation is particularly relevant for the calculation of the value-at-risk (VaR). Hence, we computed VaR for three different confidence levels, $\beta = 0.99, 0.995, 0.999$, as function of $p_G$, taking fixed $p_B=0.028$ in the Scenarios 1 ($\rho_G < \rho_B$) and 2 ($\rho_G > \rho_B$) and under two skewness settings, $\alpha<0$ and $\alpha>0$. The plots reported in Figure \ref{fig:vars_scenario} show that: 1) in all cases, the portfolio VaR increases monotonically with $p_G$, as expected due to the increasing risk profile; 2) the skewness parameter $\alpha$ has a non-negligible effect on VaR, increasingly more pronounced for higher level of confidence; 3) VaR curves rise more steeply when $\rho_G > \rho_B$, highlighting the greater sensitivity of the overall portfolio risk to the systematic factor when the green component is more sensitive to systematic shocks.  Overall, the effect of the factor loading differential becomes more pronounced when the default probability of the green component is lower than that of the brown component. In this case, a higher correlation with the common systematic factor (Scenario 2) amplifies the portfolio's sensitivity to macroeconomic shocks, resulting in a steeper increase in VaR. Conversely, when $p_G$ is considerably greater than $p_B$, this sensitivity difference tends to diminish, and the VaR levels become more comparable across scenarios, suggesting that the dominating influence of the higher marginal default risk of the green sub-portfolio overshadows the role of correlation structure.

Finally, as demonstrated in Remark \ref{rem1}, the quantile dependence on portfolio weights is linear, see Figure \ref{fig:VaRs}. Specifically, as expected, VaR decreases as a function of $\omega_G$ (the relative portfolio weight of green loans), provided that either the default probability or the factor loading of green loans is lower than the corresponding values for nongreen loans.

\begin{figure}[h]
\begin{center}
\includegraphics[width=6.5cm,height=5.5cm]{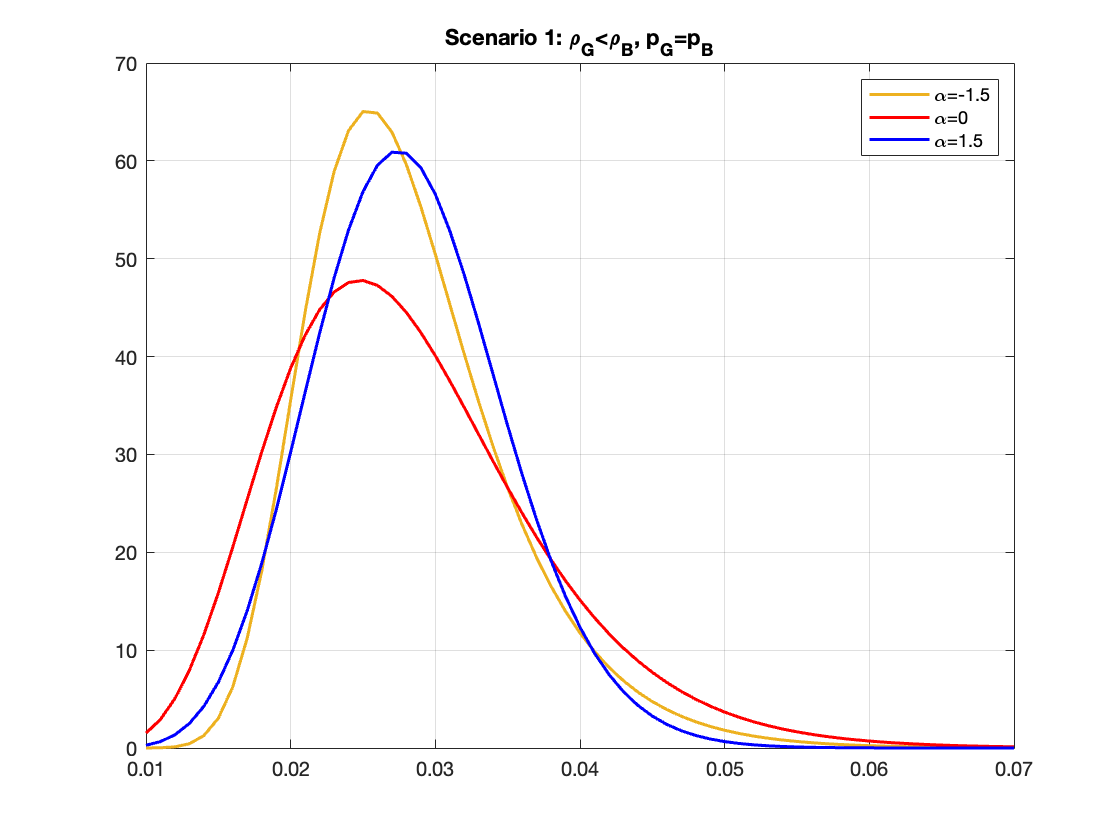}
\includegraphics[width=6.5cm,height=5.5cm]{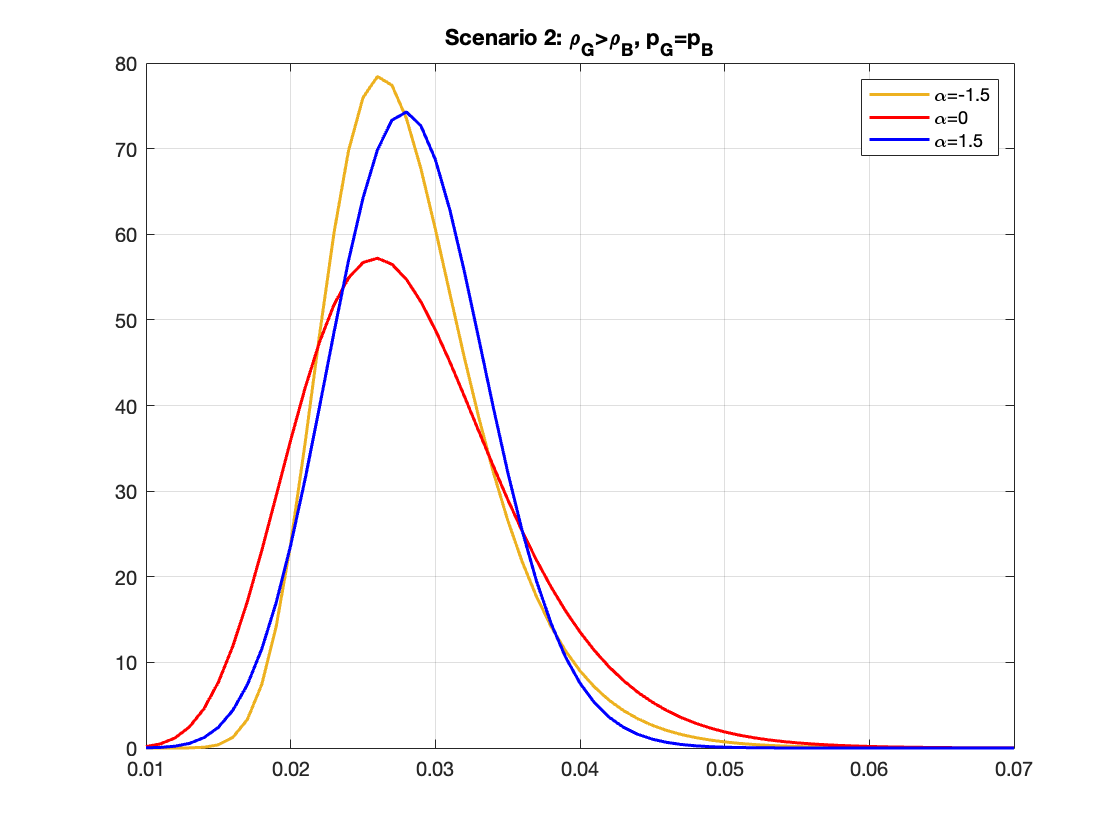}
\includegraphics[width=6.5cm,height=5.5cm]{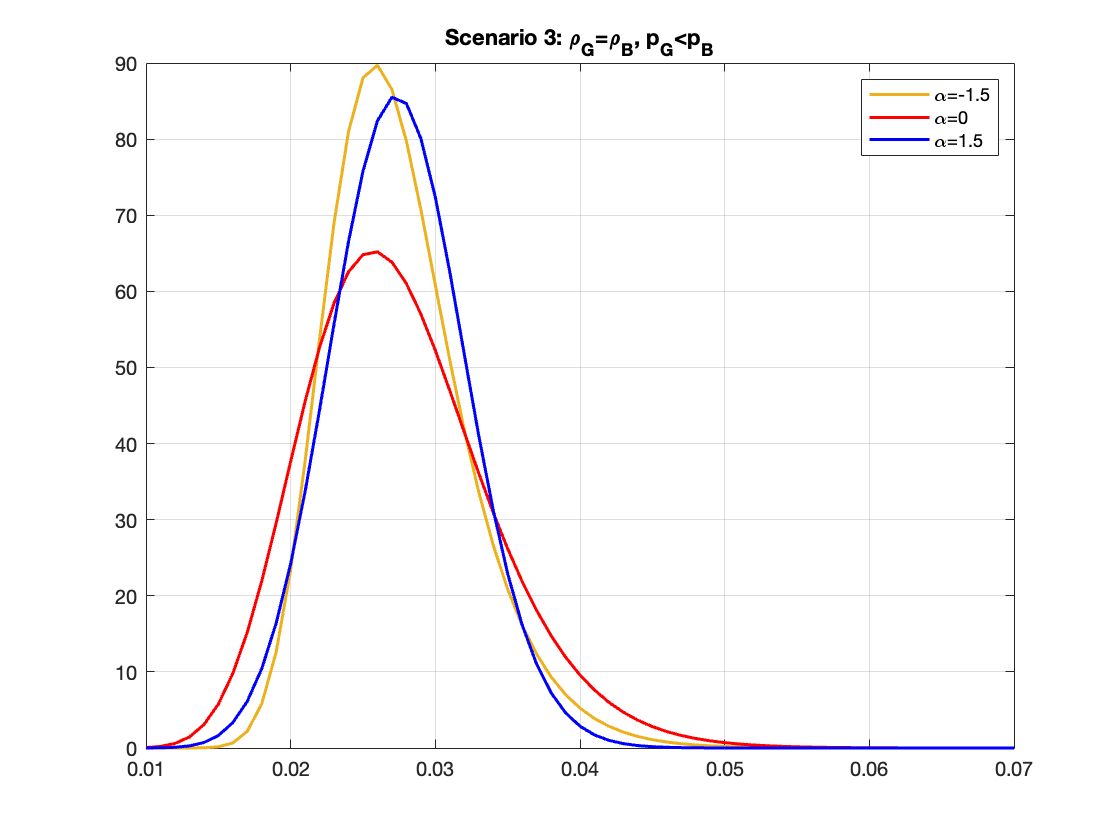}
\includegraphics[width=6.5cm,height=5.5cm]{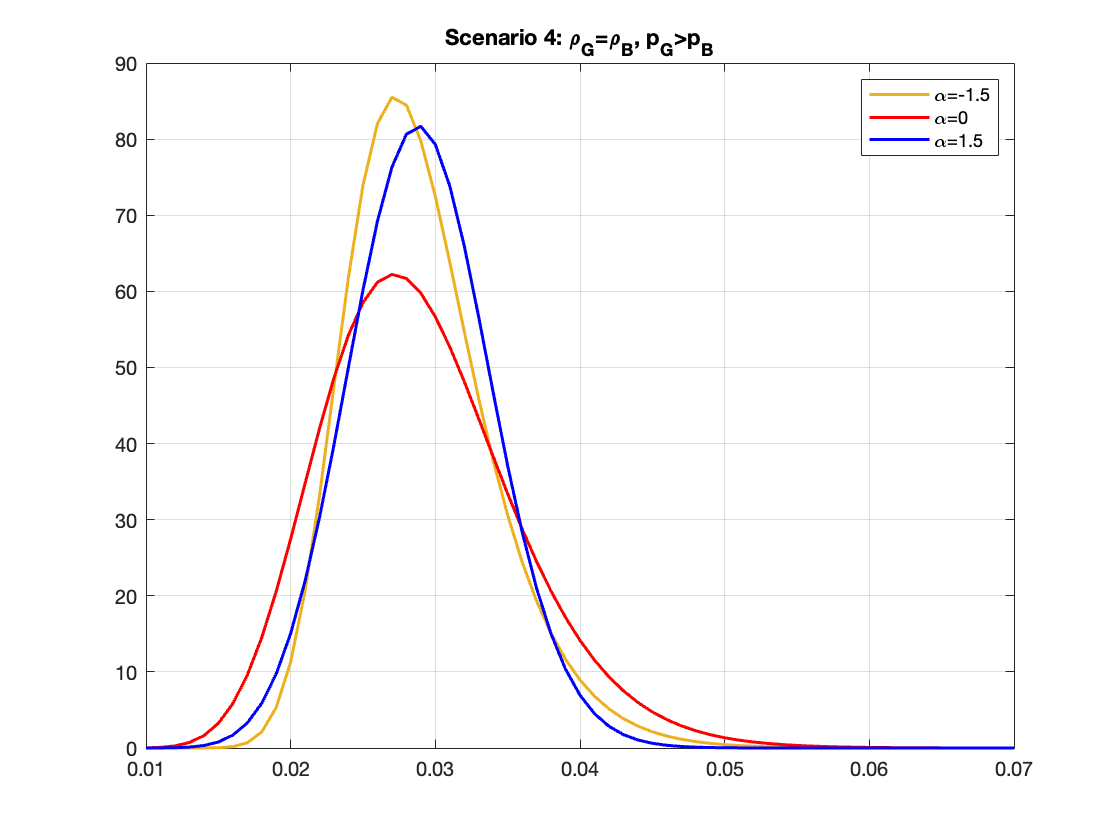}
\end{center}
\caption{The density of $L^{mix}_{\omega_b,\omega_g}$ for the 4 scenarios, and three skewness settings. The case $\alpha=0$ corresponds to the standard normal distribution assumption on the systematic risk factor.}
\label{fig_densities}
\end{figure}

\begin{figure}[h]
\begin{center}
\includegraphics[width=8.5cm,height=8cm]{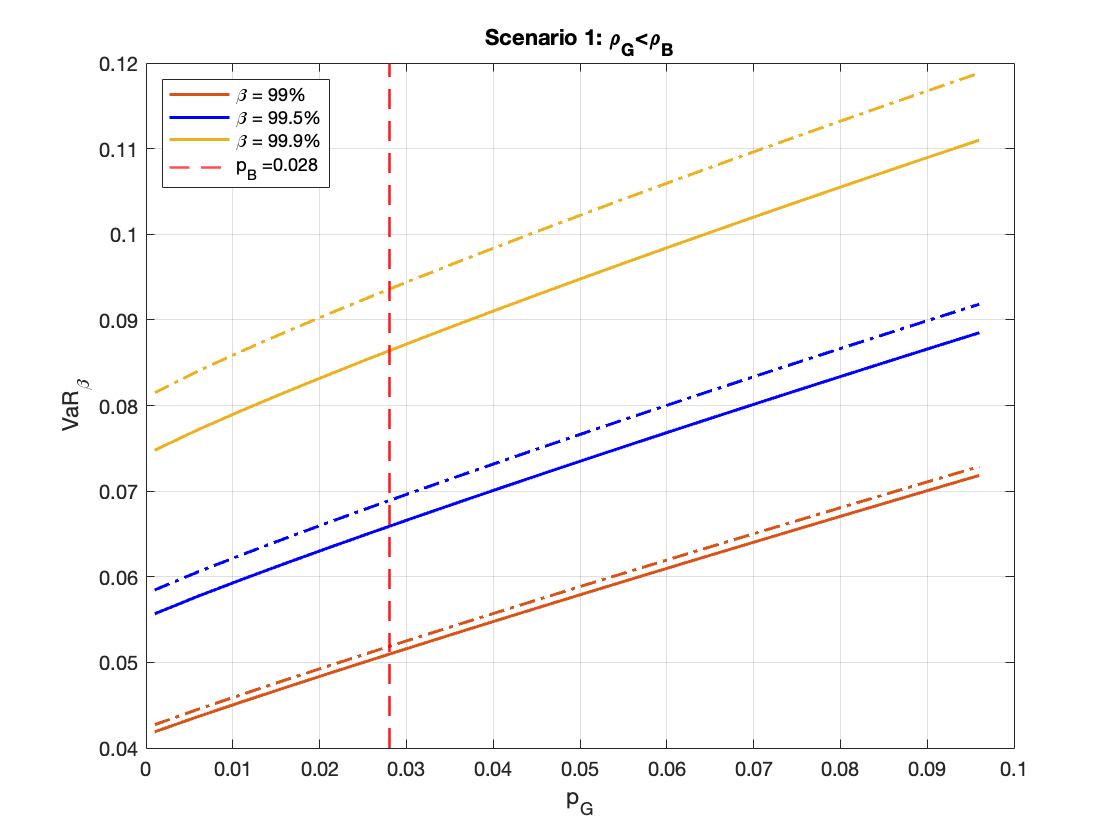}
\includegraphics[width=8.5cm,height=8cm]{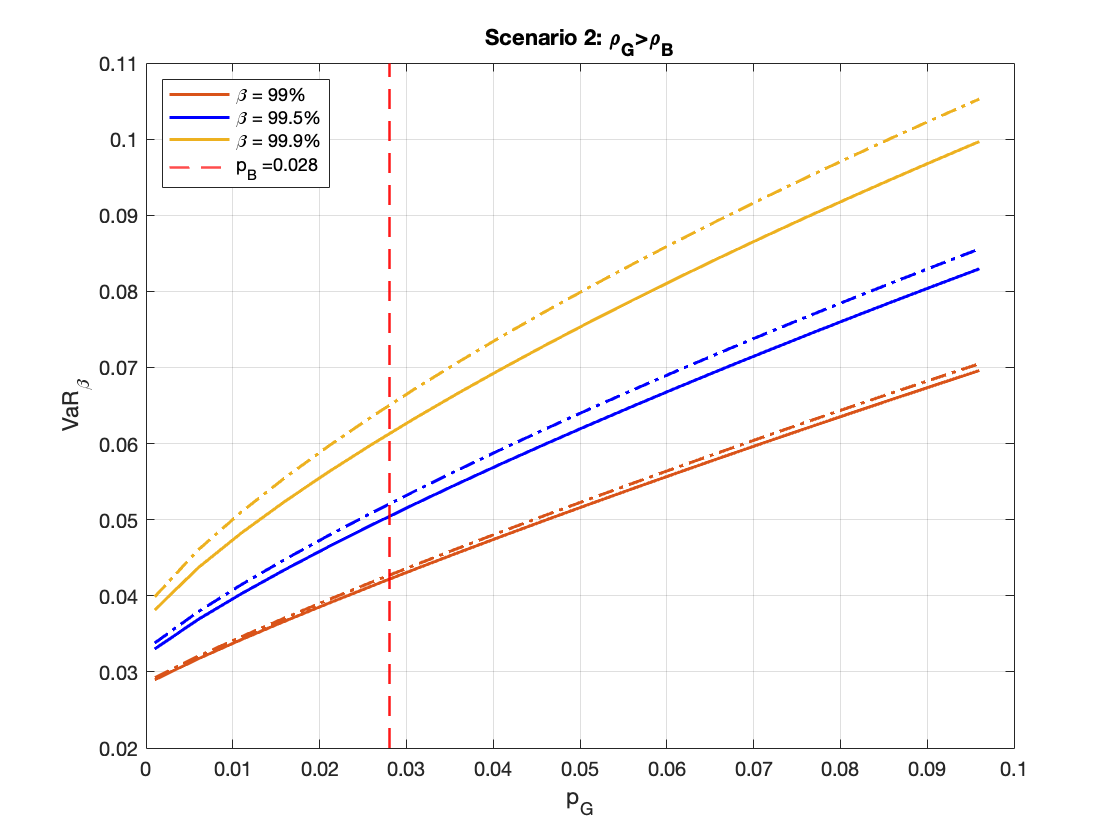}
\end{center}
\caption{VaR profiles of the green-brown portfolio as a function of the default probability of the green portfolio component, $p_G$. The default probability of the brown component is held fixed, $p_B=0.028$. The continuous lines correspond to a negative value of the skew-normal parameter, $\alpha=-0.8$, the dash-dotted line to $\alpha=0.8$.}
\label{fig:vars_scenario}
\end{figure}

\begin{figure}[h]
\begin{center}
\includegraphics[width=6.5cm,height=5.5cm]{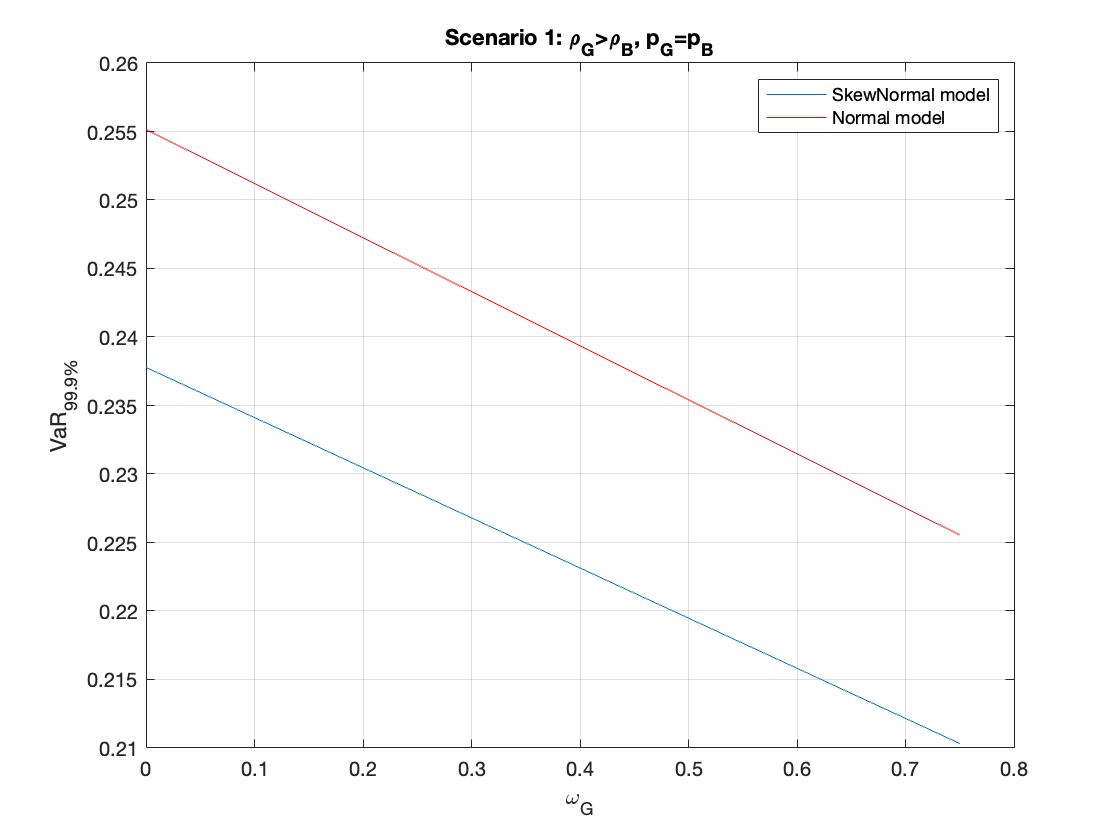}
\includegraphics[width=6.5cm,height=5.5cm]{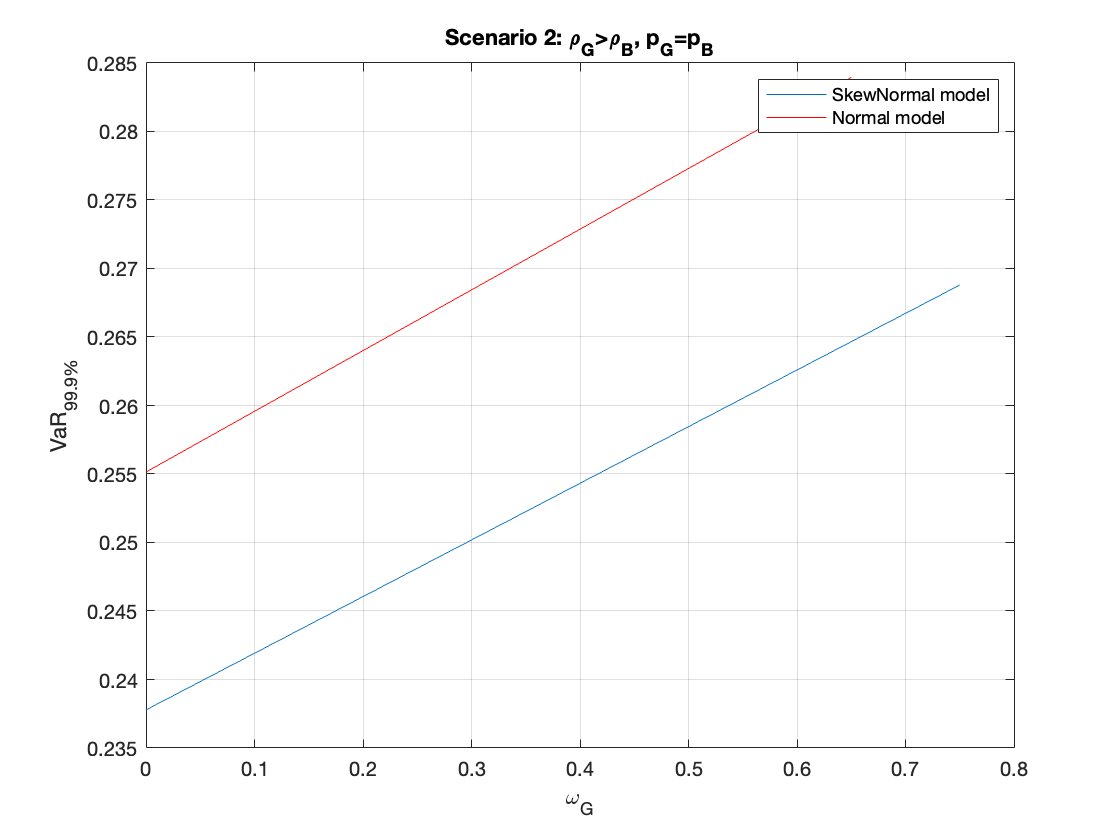}
\includegraphics[width=6.5cm,height=5.5cm]{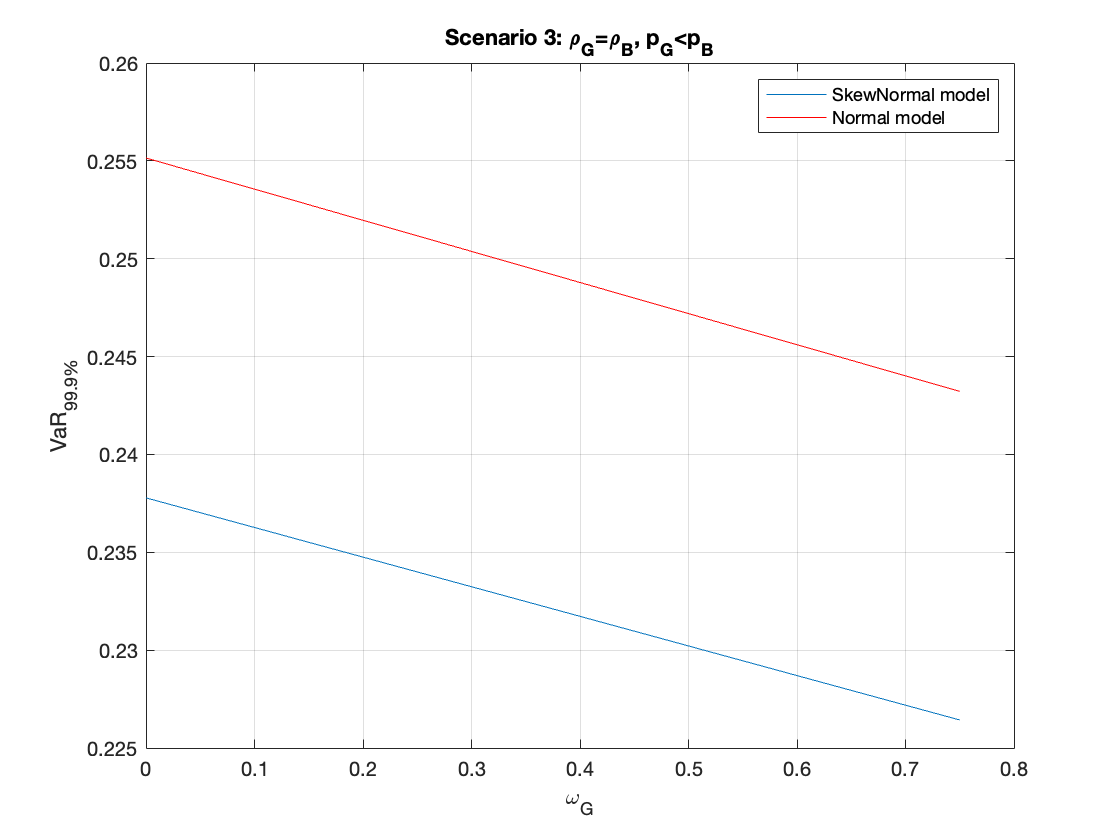}
\includegraphics[width=6.5cm,height=5.5cm]{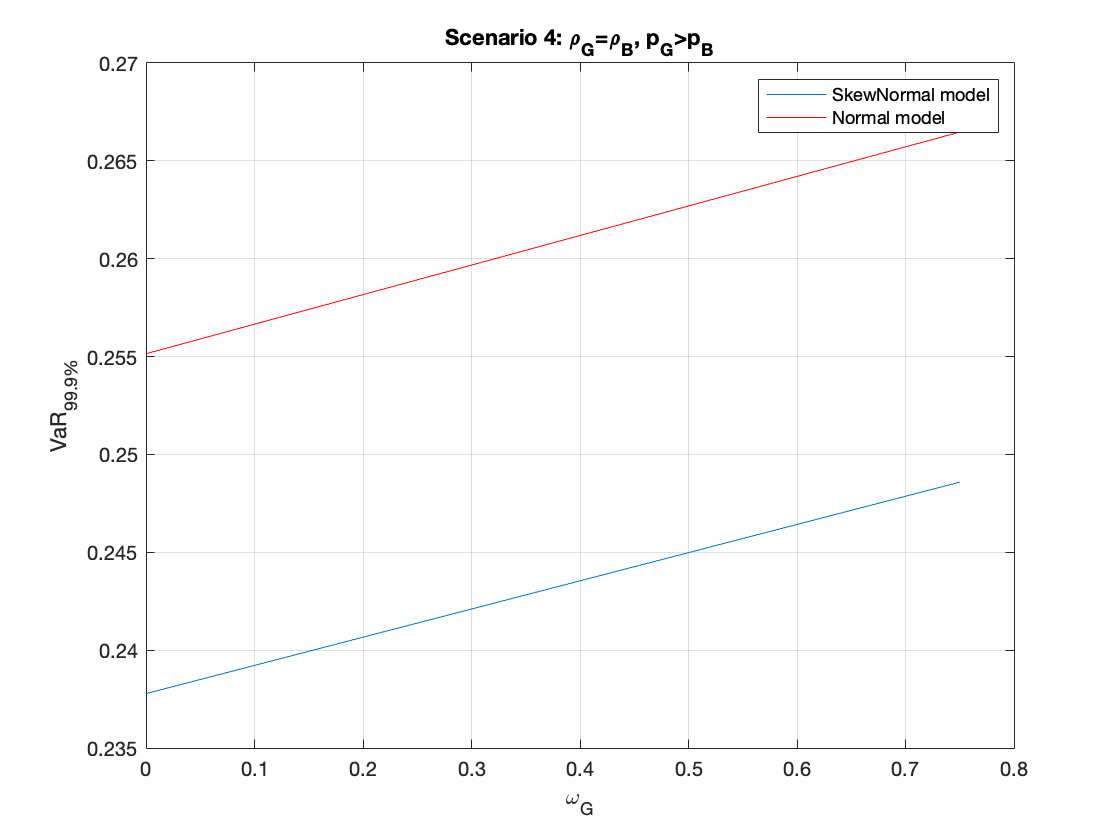}
\end{center}
\caption{VaR at level $\beta = 99.9\%$ for the green-brown portfolio as a function of the relative weight $\omega_G$, for different scenarios. The red line corresponds to the Gaussian model ($\alpha = 0$).}
\label{fig:VaRs}
\end{figure}

\section{Conclusions}

In this paper, we proposed an extension of the classical ASRF model framework to account for heterogeneous credit portfolios composed of green and brown loans. By adopting a two-factor copula structure and allowing for skew-normal specifications of the systematic risk components, our model captures asymmetric dependence patterns and differential risk sensitivities across the two sub-portfolios. We further incorporated non-uniform exposure distributions to reflect realistic portfolio compositions and derived a quadratic mean convergence result under general conditions, with the limiting loss variable reflecting the structural distinctions between the green and brown components.

Our numerical investigations confirmed the theoretical convergence and highlighted the role of key model parameters in shaping the loss distribution. In particular, we demonstrated how exposure granularity, skewness, default probabilities, and factor loadings influence the value-at-risk, with marked effects under varying portfolio configurations. These results underscore the importance of explicitly modeling green and brown credit segments when evaluating portfolio-level credit risk, especially in the context of sustainability-oriented lending.

Future research may extend the current framework in several directions. One promising avenue involves adapting the granularity adjustment technique to our setting in order to refine the LHP approximation of value-at-risk and better quantify idiosyncratic risk contributions in finite portfolios. Another line of work concerns the extension of securitization analysis and CDO pricing methodologies (see \cite{agl21}) to accommodate the skew-normal specification of systematic risk factors, thus improving the valuation in asymmetric credit environments. Finally, empirical investigations based on borrower-level data could offer valuable insight into the structural differences between the green and brown loan segments. In particular, such analyses could address whether green loan portfolios exhibit greater resilience to systematic shocks, with implications for both portfolio management and sustainable finance policy.

\section{Acknowledgments}
The Authors have been supported by the PRIN 2022 PNRR Project "The resilience of sustainable finance"(Project Number: P2022ENNYP—CUP: B53D23026490001) funded by the European Union--Next Generation EU.


\section*{Appendix 1}
\begin{Proposition} (proof of formula (\ref{skew-gaussian}))\\
We have
\begin{equation}\label{skew-normal}
{\bf{Y}}^{(N)}\sim SN({\bf{0}},{\bf{b}}{\bf{b}}'+{\bf{D}}{\bf{\Sigma}}, \frac{{\bf{\Sigma}}^{-1}{\bf{b}}}{\sqrt{1-{\bf{b}}'{\bf{\Sigma}}^{-1}{\bf{b}}}}),
\end{equation}
with ${\bf{b}}=(\beta_b{\bf{1}}^b,\beta_g{\bf{1}}^g)'$, where $\beta_a\equiv \delta_a\rho_a$, ${\bf{D}}=diag(\sqrt{1-\beta^2_b}{\bf{1}}^b,\sqrt{1-\beta^2_g}{\bf{1}}^g)$, and ${\bf{\Sigma}}$ is a partitioned covariance matrix specified below.
\end{Proposition}
\begin{proof} 
Recall that for $h=1,\ldots,N_a$ we have
$$
Y^{(N_a)}_h=\rho_a\sqrt{1-\delta_a^2}\;X_1 +\rho_a\delta_a\; X_2+\sqrt{1-\rho_a^2}\;Z^{(N_a)}_{h}
=\beta_a X_2+\sqrt{1-\beta_a^2}V^{(N_a)}_h,
$$
where $\beta_a\equiv \rho_a\delta_a$ and we have set
$$
V^{(N_a)}_h\equiv \sqrt{\frac{\rho_a^2-\beta_a^2}{1-\beta_a^2}}X_1+\sqrt{\frac{1-\rho_a^2}{1-\beta_a^2}}Z^{(N_a)}_h.
$$
The random vector ${\bf{V}}^{(N)}=({\bf{V}}^{(N_b)},{\bf{V}}^{(N_g)})$ is gaussian distributed with $N(0,1)$ marginal laws and $N\times N$ covariance matrix ${\bf{\Sigma}}$ partitioned in the following way
$$
{\bf{\Sigma}}= \left[ {\begin{array}{cc}
   {\bf{\Sigma}}^{bb} & {\bf{\Sigma}}^{bg} \\
   {\bf{\Sigma}}^{gb} & {\bf{\Sigma}}^{gg} \\
  \end{array} } \right],
$$
where ${\bf{\Sigma}}^{bb}$ is ${N_b\times N_b}$, ${\bf{\Sigma}}^{gg}$ is ${N_g\times N_g}$, $ {\bf{\Sigma}}^{bg}$ is ${N_b\times N_g}$ and ${\bf{\Sigma}}^{gb}=({{\bf{\Sigma}}^{bg}})'$. The coefficients of the previous matrices are as follows: $\Sigma_{ij}^{bb}=\left(\frac{\rho_b^2-\beta_b^2}{1-\beta_b^2}\right)$ for $i\neq j$ and 
$\Sigma_{ii}^{bb}=1$, $\Sigma_{ij}^{gg}=\left(\frac{\rho_g^2-\beta_g^2}{1-\beta_g^2}\right)$ for $i\neq j$ and 
$\Sigma_{ii}^{gg}=1$, finally $\Sigma_{ij}^{bg}=\left(\frac{\rho_b^2-\beta_b^2}{1-\beta_b^2}\right)^{1/2}\left(\frac{\rho_g^2-\beta_g^2}{1-\beta_g^2}\right)^{1/2}$.
Clearly we have
\begin{equation}\label{conv}
{\bf{Y}}^{(N)}=({\bf{Y}}^{(N_b)},{\bf{Y}}^{(N_a)})={\bf{b}}X_2+{\bf{D}}{\bf{V}}^{(N)}
\end{equation}
for ${\bf{b}}=(\beta_b{\bf{1}}^b,\beta_g{\bf{1}}^g)'$ and ${\bf{D}}=diag(\sqrt{1-\beta^2_b}{\bf{1}}^b,\sqrt{1-\beta^2_g}{\bf{1}}^g)$. The representation of ${\bf{Y}}^{(N)}$ as sum of two random vectors allows to determine its distribution by convolution which,  by routine calculations, leads to (\ref{skew-gaussian}), see e.g.  Azzalini (2005), pg.173-174.
\end{proof}

\end{document}